\documentclass[a4paper,onecolumn,11pt]{quantumarticle}
\pdfoutput=1
\usepackage[utf8]{inputenc}
\usepackage[english]{babel}
\usepackage[T1]{fontenc}
\usepackage{amsmath}
\usepackage{hyperref}

\usepackage{tikz}  
 
\usepackage{fullpage}
\usepackage{amsfonts,amssymb}
\usepackage{amsthm}
\usepackage{tikz}
\usetikzlibrary{calc}
\usetikzlibrary{arrows}
\usepackage{tikz-3dplot}
\usepackage{color}  
\usepackage{hyperref} 
\usepackage{bm}
\usepackage{float} 
\usepackage{appendix} 
\usepackage{lscape}
\usepackage{cite}
\usepackage{mathrsfs}
\usepackage{appendix}
\usepackage{setspace}
\usepackage{mathtools}

\theoremstyle{definition}
\newtheorem{definition}{Definition}[section]
\newtheorem{lemma}[definition]{Lemma}
\newtheorem{theorem}[definition]{Theorem}

\begin{document}
	
	\title{Tsirelson's Bound and the Quantum Monogamy Bound from Global Determinism}
	
	\author{Emily Adlam} 
	
		\affiliation{The University of Western Ontario}
	
	\maketitle
	
	\abstract{ We demonstrate that in a globally deterministic universe, all spatiotemporally symmetric processes must obey counterfactual parameter independence. We show that  the Tsirelson bound can be derived from counterfactual parameter independence. We show that the quantum monogamy bound can also be obtained from global determinism, and then then use these results to propose a novel solution to the horizon problem. We also explain how global determinism relates to contextuality in quantum mechanics. }

	\section{Introduction} 
	
	In 1964 Bell showed that quantum mechanics violates a certain inequality which must be satisfied by any local theory\cite{Bellinequality}, and thus set off a conversation around non-locality which continues to this day. Broadly speaking, there are two opposing schools of thought about how we should respond to Bell's theorem. The first says that we should find a way to rescue locality by denying one of the assumptions that goes into the theorem. However, this requires us to accept behaviours which may seem just as surprising as nonlocality, such as superdeterminism\cite{10.3389/fphy.2020.00139}. The alternative is to make our peace with the existence of genuine non-local correlations. However, we are still left with a puzzle: if the world is generically non-local, why are there such strong limits on the ways in which non-locality can be manifested? For example, why can't we use non-locality to send signals faster than light? 

In ref \cite{QMG} we set out a partial answer to this puzzle. We observed that if the world is generically non-local, then the usual reasons for thinking that quantum mechanics is probabilistic do not apply, since the outcomes of measurements could depend on global variables as well as the local ones to which we have access. Thus it could be the case that the universe is deterministic in a generalized sense on a global scale. We then showed that signalling is not compatible with global determinism. The reason for this is quite simple: if signalling were possible, then Alice could send a faster-than-light message to Bob, and Bob could send the content of this message as another faster-than-light message back to Alice so that it arrived just in time for her to use it as the content of her original message to Bob. This would create a closed causal loop where the content of the message would effectively have come out of of nowhere, i.e. it would not be determined by anything. Thus the hypothesis of global determinism explains the fact that we can't send signals faster than light (and equivalently, the fact that we can't send signals faster than light can be regarded as evidence for the hypothesis of global determinism). In ref \cite{QMG} we formalised this intuition using the language of entropy and mutual information. 

However this is not the end of the story, because no-signalling does not suffice to define the set of non-local correlations which can be achieved within quantum mechanics. This fact emerged via the work of Popescu and Rohrlich on non-local boxes\cite{Rohrlich} – that is, sets of hypothetical devices which accept inputs and produce outputs which may be correlated in any mathematically describable way. Popescu and Rohrlich described a specific pair of boxes, now known as  ‘PR boxes,’ which cannot be used to perform signalling but which can be used to produce a value of $4$ for a particular function of expectation values known as the CHSH quantity. But it has been shown by Tsirelson that no arrangement of quantum mechanical systems can produce a value greater than $2\sqrt{2}$ for the CHSH quantity\cite{Tsirelson}, so the PR boxes exhibit a stronger form of non-locality than quantum mechanics. Thus one might hope that it is possible to use global determinism to derive not only the no-signalling bound but the stronger quantum bound. In this paper we demonstrate that this is indeed the case - the Tsirelson bound can be derived from a principle known as information causality, which can be derived from global determinism. 

 Moreoever, a similar situation crops up in the study of quantum monogamy: it has been shown that any no-signalling theory must obey a monogamy bound, but the monogamy bound obeyed by quantum mechanics is stronger than the bound implied by no-signalling. We show that this quantum bound can likewise be derived from global determinism, since it follows from no-signalling and information causality. We describe a possible application of this result to the black hole horizon problem. Finally, we discuss the relationship between global determinism and contextuality.

\section{Background}

\subsection{Global Theories}

It is common in physics to suppose that the universe works like a computer, taking an initial state and evolving it forward in time in a temporally local way, so that events at a given time can depend only on other events at the same time. But Bell's theorem tells us that the laws of nature involve some sort of nonlocality \cite{Bellfree}, and special relativity indicates that if the laws of nature are spatially non-local then they are likely to also be temporally non-local\cite{adlamspooky}, and therefore we have good reasons to reconsider this assumption of temporal locality. 

Although the mainstream literature on quantum mechanics is still dominated by temporally local approaches, in recent years a variety of interesting work has been done on non-standard temporal pictures. For example, there has been growing interest in retrocausal approaches to quantum theory\cite{Goldstein_2003, PRICE_1994, Priceretro}, including a proof by Pusey and Leifer demonstrating that if quantum mechanics obeys a certain sort of time-symmetry then it must exhibit retrocausality\cite{PuseyLeifer} and a model due to Wharton which suggests a natural resolution to the quantum reality problem using the `all-at-once'-style analysis of action principles\cite{Whartoninformation}. In such models it is supposed that the laws of nature apply to the whole history of the universe all at once, so that events at a given time can influence events in both the past and future without the need for any mediating state. In this paper we adopt this global approach and investigate what local observations can tell us about the structure of the global laws.  Specifically, we consider the hypothesis that the global laws are deterministic in a generalized sense, and show that this hypothesis can be used to explain some characteristic features of quantum mechanics. Thus by inference to the best explanation, this result provides some support for the hypothesis of global determinism and more broadly for the global approach. 
 
What exactly do we mean by `deterministic in a generalized sense'? Well, recall that the usual Laplacean definition of determinism is intended for a temporally local setting: a theory is said to be deterministic if the state on a given time-slice fully determines all events to the future of that time-slice. But there is no reason to adopt this sort of arbitrary temporal constraint in a theory where the laws of nature apply all-at-once to the whole of history, so in this context we need a less restrictive definition of determinism. Thus we will say that a \emph{globally deterministic} theory is one where the laws of nature prescribe a unique course of history, though it may not be possible to infer the course of history from the data on any single time-slice. 

Global determinism is entirely consistent with the apparently probabilistic nature of quantum mechanics, since quantum mechanics is based on our limited observations from a local standpoint. An analogy may be helpful. If a completed sudoku square were to be revealed  column by column to an observer who could not see the whole square at once, each column would appear to be related to the previous one in a probabilistic way; there would be obvious patterns of dependence which could be described by probabilistic rules, but there would not usually be enough information available to determine the next column exactly. Nonetheless, if the grid has been set up properly then when one considers the whole puzzle at once it is clear that the solution is unique. Similarly, in the global approach the apparently probabilistic nature of quantum theory may be a consequence of  our inability to see the whole picture: if we were able to consider the whole universe at once it would be clear that the `solution' is unique. 

Note that global determinism is not the same as superdeterminism, which is a research programme in quantum mechanics concerned with the possibility of restoring locality to quantum mechanics by postulating models which  violate the principle of statistical independence which is used in the proof of Bell's inequality. Despite the name, superdeterministic models do not actually have to be deterministic, and when they \emph{are} deterministic this is usually supposed to be determinism in the standard Laplacean sense. Global determinism is significantly more general than this idea: certainly it would be possible for a model to satisfy both global determinism and superdeterminism, but the two do not have to go together. 

\subsection{Information Causality \label{entropy}}

Non-locality for bipartite systems is often measured using the $CHSH$ quantity. Suppose we may perform two different measurements $M^0, M^1$ on two different systems $A$ and $B$, each measurement having two possible outcomes labelled by $1$ and $-1$. We define the $CHSH$ quantity for this arrangement by 	 the following sum of expectation values:

 \[ CHSH_{AB}   =   \langle A B \rangle_{ M_A^0, M_B^0}  + \langle A B \rangle_{  M_A^0, M_B^1} +\langle  A B \rangle_{  M_A^1, M_B^0}- \langle A B\rangle_{ M_A^1, M_B^1} \] 

For classical systems, this quantity can be no larger than two, while for general non-signalling systems such as the $PR$ boxes it can be as large as 4. For quantum systems it is bounded above by $2 \sqrt{2}$, a value known as the \emph{Tsirelson bound}. Thus if we accept that Bell's theorem proves the existence of non-locality, we are left with  an important question about why quantum mechanics is not as non-local as the maximally non-local non-signalling theory.

 In ref \cite{Pawlowski} it was shown that the Tsirelson bound can be derived from an information theoretic principle known as information causality. This principle is most easily set out in the context of a game in which Alice is in possession of a string of bits $A = a_0 a_1 ... a_n$ and is allowed to send Bob a classical message $c$, and then Bob is given an integer $m$ chosen at random from $[0 ...n)$ and is required to guess the value of the $m^{th}$ bit of Alice's string. Information causality can then be understood as the requirement that the total amount of information accessible to Bob, across all the measurements that he could possibly choose to perform, is no greater than the amount of information in Alice's message $c$. This criterion is expressed formally in terms of classical Shannon entropies and mutual informations:

\[ \sum_x I( a_x : g c | m = x ) \leq H(c)  \]

This principle has a certain level of intuitive plausibility because it looks like a generalisation of the no-signalling principle, but it is important to reinforce that in quantum mechanics the measurements corresponding to different values of $m$ cannot all be performed simultaneously, and therefore this principle constrains not only the amount of information that Bob can actually obtain, but also the amount the he could \emph{counterfactually} obtain across a set of measurements that he cannot perform simultaneously. Information causality is therefore not implied by relativistic causality (which deals only with actual facts, not counterfactual ones) and indeed on reflection it seems a rather surprising constraint - why should the universe care, so to speak, about counterfactual measurements that cannot actually be performed? After all, as demonstrated by the PR boxes, it is possible to come up with correlations that violate this constraint which seem at least superficially reasonable.

In the subsequent literature it was observed that information causality is closely connected to features of the quantum mechanical von Neumann  entropy, and indeed there exist several derivations of it from  entropic principles. Ref \cite{Al_Safi_2011} shows that information causality can be derived from the existence of a measure of entropy for systems which reduces to the classical (Shannon) entropy if the system is classical and which satisfies the condition that for any two systems $X, Y$, if a transformation is performed on $Y$ alone then the change in entropy for $X, Y$ together is greater than or equal to the change in entropy for $Y$ alone. Ref \cite{Barnum_2010} shows that information causality can be derived from the existence of a measure of entropy such that the entropy associated with measurements is the same as the entropy associated with preparations, provided that this entropy is strongly subadditive   and satisfies the Holevo bound\cite{Holevo1973SomeEO}. Ref \cite{Dahlsten_2012} shows that information causality can be derived from the existence of a measure of entropy which satisfies the data processing inequality. Clearly  these proofs are all latching onto the same underlying fact - information causality must be satisfied if it is possible to write down a measure of entropy which behaves like the classical Shannon entropy in certain sorts of situations. Moreoever, because the Shannon entropy and its properties are so familiar and intuitive, these results initially seem quite compelling. 

However, it is important to be careful with such analogies. Recall that the Shannon entropy is designed to describe a random variable and thus our intuitions about it are based on our understanding of the behaviour of random variables, so strictly speaking it is incorrect to talk about the Shannon entropy of a \emph{system}. However, it turns out that we \emph{can} sensibly define the Shannon entropy of a classical system, for  one simple reason: measurements on classical systems are non-destructive. This means that for any state of a classical system we can always perform all possible measurements and thus the system can be described by a random variable which takes values indexed by all possible combinations of measurement outcomes. Thus from a mathematical point of view a classical system is equivalent to a random variable, and therefore it is not surprising that we can describe classical systems using an entropy which obeys entropic relations like the data processing inequality, even though the Shannon entropy is strictly speaking applicable to random variables rather than systems. 

But this is \emph{not} true for quantum systems, because quantum measurements usually cannot be performed simultaneously, and a system for which some observables are not simultaneously measurable need not behave like a classical variable. For example, each of refs \cite{Al_Safi_2011, Barnum_2010, Dahlsten_2012} uses their chosen entropic principles to derive some variant on the following relation (where the bits $a_x$ of $A$ are independent, and where $B$ represents some quantum system which Bob may measure in order to determine his guess for the desired bit value): 

\[   I ( A : B c  ) \geq \sum_x I ( a_x : B c )   \]

Since the entropy of a system is usually understood to represent something like the minimum Shannon entropy over measurement outcomes (minimizing over all possible non-trivial measurements)\cite{Barnum_2010}, it is straightforward to see that if $b$ is the outcome of the measurement $m$ performed on $B$, we can obtain the following equation, from which information causality can easily be derived:
  
\begin{equation} \max_q I ( A : b c | m = q) \geq \sum_x I ( a_x : b c | m = x) \label{dp} \end{equation}

That is, the maximum information about $A$ that we can obtain from a single measurement on $B$ together with the classical message $c$ must be greater than the sum of the information we could obtain about the bits of $A$ in a set of distinct measurements on $B$ together with the classical message $c$. Now, it is obvious that this must be true in a system where all  measurements can simultaneously be performed, since we can simply combine all the measurements $m = x$ into a single measurement which necessarily gives us all the same information about the bits of $a$ as we could have obtained by performing the measurements separately. But in a theory with incompatible measurements equation \ref{dp} need not hold, because one could in principle obtain non-overlapping information from each of the incompatible measurements on the right hand side of the equation. So although equation \ref{dp} does in fact hold in quantum mechanics, the reason it holds cannot be the same as the reason it holds in classical mechanics. 

This means that trying to explain information causality by appealing to the existence of an entropy measure which behaves similarly to the Shannon entropy really only serves to restate the fundamental problem in a more mathematical language: we are left with an essentially equivalent puzzle about why there should be any such measure in a theory which incorporates incompatible measurements. Certainly, equation \ref{dp} is a is a necessary condition for the possibility of employing an entropy-like measure which obeys a data-processing inequality, which entails that 
if we take the data processing inequality or the existence of an entropy-like measure as a basic postulate we can derive equation \ref{dp}. But this seems to get things the wrong way round: unless we are willing to  reify the von Neumann entropy and regard it as a fundamental object of the theory\footnote{Although it would in principle to be possible to consider the von Neumann entropy to be fundamental, this route has not been taken by any mainstream approach to the interpretation of quantum mechanics. Moreover it would be a controversial position, since that entropy is usually understood as a measure of ignorance or uncertainty, and most people think that ignorance or uncertainty must be \emph{about} something\cite{timpson2008philosophical}, which suggests that entropy can't be fundamental. This view might perhaps find a more natural home within the `It From Bit'\cite{Wheeler1989-WHEIPQ, D_Ariano_2016} approach to the interpretation of quantum mechanics, which suggests that information should be regarded as fundamental, but it is not clear that this approach has yet been formulated in a philosophically coherent way.} then the existence of an entropy-like measure cannot explain why quantum mechanics obeys equation \ref{dp} -  rather it is the fact that quantum mechanics obeys equation \ref{dp}  and others like it which makes the entropic representation possible in the first place! Rather than deriving equation \ref{dp} and thus information causality from entropic principles, we should see equation \ref{dp} as a precondition for the possibility of using entropic terminology, which means we still need an explanation for equation \ref{dp}, and more  generally, we need an explanation for the fact that under certain circumstances quantum systems behave as if they can be described by classical variables, even though the existence of incompatible measurements blocks the simple explanation available in the case of classical systems.

Note that unless otherwise stated, the entropies and mutual informations that feature in the mathematical results and proofs in this paper are \emph{classical} (Shannon) entropies, and therefore they automatically obey all the usual rules of classical information theory, such as strong subadditivity and the data processing inequality.

\subsection{The Process Framework} 

In this paper we will be dealing with scenarios involving non-standard causal ordering, and so we will work within the process framework developed by Oreshkov et al\cite{Oreshkov, Oreshkov2, Shrapnel_2018}. Here, an experiment is understood to consist of \emph{local labelled regions} where one can perform \emph{local controlled operations} $J^a, J^b, J^c ... $ that can be associated with \emph{outcomes} $a, b, c ... $. We will sometimes refer to these regions as \emph{non-local boxes} in order to make contact with the literature on non-local correlations; each non-local box takes an input (the local controllables for the region) and produces an output (the outcome for the region).

In addition to the local controlled operations, the outcomes for a set of local labelled regions may depend on the \emph{environment}, which may include `global properties, initial states, connecting mechanisms, causal influence or global dynamics.'\cite{Shrapnel_2018}  By definition, the environment is independent of the local controllables. (If this is not the case, then we can always redefine the local controllables to include whatever part of the environment fails to be independent). 

We then define a process  as an equivalence class of environments, where one environment $W$ is equivalent to another environment $W'$ if $p(a, b, c ... | J^a, J^b, J^c ..., W) = p(a, b, c ... | J^a, J^b, J^c ..., W') $ for any possible set of local controllables $J^a, J^b, J^c ... $ and outcomes  $a, b, c ... $. Thus it is part of the definition of a process that it can (in principle) be reproduced as many times as we like, and that all instances of the process will exhibit the same statistics, i.e. there will be stable probabilistic relationships between the inputs and outputs. For example, performing measurements on two particles in a maximally entangled state is a process, with the choices of measurement directions as the local controllables and the results of the measurements as the outcomes. If we repeat this experiment many times on different maximally entangled states, we will see stable probabilistic relationships between the measurement directions and the measurement results, demonstrating that we have indeed defined an equivalence class of environments.  The fact that the state concerned is maximally entangled is part of the definition of the process; a different choice of entangled state would lead to different probabilistic relationships and so measuring particles in a different entangled state would be a different equivalence class of environments and thus would constitute a different process. 

As noted in ref \cite{Shrapnel_2018}, there is some freedom in where to put the line between `process' and `local controllables.' For example, we can consider `measuring two particles in a maximally entangled state' to refer to a process, with the two choices of measurements as local controllables; but we could also fix the measurement to be made on one of the particles, so it becomes part of the process and the only local controllable is the measurement on the other particle; or we could fix both measurement directions, so they are both part of the process and there are no local controllables. Since the two measurement directions are independent, these choices still satisfy the requirement that the environment and process should be independent of the local controllables. 

However, if we say that local controllables may be incorporated into the process whereas outcomes may not, we are introducing a temporal asymmetry into the definition of a process, and one main purpose of moving to the process framework is to eliminate such temporal asymmetries as far as possible. Therefore we will take it that there is also freedom about where to put the line between `process' and `outcomes.' For example, in the case of measuring two entangled particles we could fix the \emph{outcome} of the measurement on one particle, leading to a process with only one outcome, or we could fix both outcomes, leading to a process with no outcomes.

It is common to study features of  quantum mechanics using the \emph{ontological models} framework, where it is supposed that correlations between preparations and measurement outcomes are mediated by an \emph{ontic state}. However, the standard ontological models framework presupposes a fixed causal order and is therefore unsuitable for the process framework. We will therefore use a generalized ontological models framework, similar to the one defined in \cite{Shrapnel_2018}(see Appendix \ref{cs} for some further comments).  Here, we allow that in addition to the local controllables and the environment, the outcomes of an experiment may depend on an  \emph{ontic variable} $\Omega$. The ontic variable is defined so as to capture physical properties of the world that remain invariant under local operations - that is, spatiotemporally local experimenters such as ourselves are not able to influence the value of the ontic variable in any way. For example, the ontic variable could include hidden variables or global facts which can't be observed or altered by local observers. The independence of the ontic variables will turn out to be essential to our arguments, because in a composition as in fig \ref{figsig}, we can set the local controllables of one process equal to the outcomes of the other process, but since we can't influence the ontic variable we can't make the ontic variable for one process depend on the outcomes of the other process. The ontic variables are therefore independent of any composition that we might perform, and thus they provide a fixed baseline of information about the outcomes of the process which is independent of choices about experimental design.

Note that the `environment' $W$  includes information about the spatiotemporal locations of the local labelled regions. Thus we can set out an important classification for processes, based on terminology introduced in ref \cite{QMG}:

 \begin{definition} 
 	
 	A process is \emph{spatiotemporally symmetric} iff it is invariant under any permutation of the spatiotemporal locations of the local labelled regions.

 \end{definition} 

In this article we will deal only with bipartite spatiotemporally symmetric processes, so the definition simply comes down to saying that the spacetime locations of the two parts of the process can be swapped without changing the operational statistics. For example, the process in which Alice and Bob each perform a measurement on one particle from a pair of entangled particles is spatiotemporally symmetric, because if we swap the spacetime locations of the measurements (including the particles, the measuring devices and Alice and Bob themselves) the correlations observed will be exactly the same. On the other hand, making a telephone call is not a spatiotemporally symmetric process, because the speaker must be in the past light-cone of the listener and so if we swap the spatiotemporal positions of the speaker and the listener we won't observe the same correlations.

The distinction between these two sorts of processes will also turn out to be crucial to our argument, because there are certain sorts of processes in the universe which \emph{can} be used to send signals - for example, telephone calls, letters and semaphore - and so the proof that global determinism prohibits signalling had better not apply to those sorts of processes. But in fact, the proof that global determinism prohibits no-signalling requires that the process in question can be composed to form a loop as shown in fig \ref{figsig}, and obviously we can't make loops like this using telephone calls, letters or semaphore, because they are all spatiotemporally asymmetric - the message must be sent in the past light-cone of the point at which it is received. On the other hand we \emph{can} make loops of this kind using quantum processes, since they are spatiotemporally symmetric. Indeed, it seems that the set of spatiotemporally symmetric processes is exactly the set of non-signalling processes, at least according to our current understanding of physics. 

Note that if we compose a spatiotemporally symmetric process with another spatiotemporally symmetric process, it is not guaranteed that the composed process will remain spatiotemporally symmetric. For example, in the composition shown in fig \ref{figsig}  two spatiotemporally symmetric processes are composed to form a process which is no longer spatiotemporally symmetric, since we are using the outcome of one part of the process to fix the local controllables for another part of the process so we can no longer    freely permute the locations.

\begin{figure}
	\centering
	\begin{tikzpicture}[scale=0.7]

	\coordinate (a) at (0,0);
	\coordinate (b) at (0,2);
	\coordinate (c) at (2,2);
	\coordinate (d) at (2,0);

	\draw[red] (a) -- (b);
	\draw[red] (b) -- (c);
	\draw[red] (c) -- (d);
	\draw[red] (d) -- (a);

	\coordinate (e) at (4,0);
	\coordinate (f) at (4,2);
	\coordinate (g) at (6,2);
	\coordinate (h) at (6,0);

	\draw[black] (e) -- (f);
	\draw[black] (f) -- (g);
	\draw[black] (g) -- (h);
	\draw[black] (h) -- (e);
	
	\coordinate (i) at (4,4);
	\coordinate (j) at (4,6);
	\coordinate (k) at (6,6);
	\coordinate (l) at (6,4);

	\draw[red] (i) -- (j);
	\draw[red] (j) -- (k);
	\draw[red] (k) -- (l);
	\draw[red] (l) -- (i);
	
	\coordinate (m) at (0,4);
	\coordinate (n) at (0,6);
	\coordinate (o) at (2,6);
	\coordinate (p) at (2,4);

	\draw[black] (m) -- (n);
	\draw[black] (n) -- (o);
	\draw[black] (o) -- (p);
	\draw[black] (p) -- (m);

	\draw[red, ->] (1,2) -- (1,2.9);
	\draw[black, ->] (1,3.1) -- (1,4);
	
	\draw[black, ->] (5,2) -- (5,2.9);
	\draw[red, ->] (5,3.1) -- (5,4);
	
	\node[below, red] at (0.5,2.8) {$O^x$};
	\node[below] at (5.5,2.8) {$O^y$};
	\node[below] at (0.5,3.8) {$I^y$};
	\node[below, red] at (5.5,3.8) {$I^x$};
	
	\draw[red, <->, line width=2pt] (c) -- (i);
	\draw[black, <->, line width=2pt] (p) -- (f);

	\node[below] at (1,1.25) {$X_1$};
	\node[below] at (5,5.25) {$X_2$};
	\node[below] at (1,5.25) {$Y_2$};
	\node[below] at (5,1.25) {$Y_1$};
	\end{tikzpicture}	
	\caption{Schematic diagram of the composition of two processes used to prove that in a globally deterministic universe all spatiotemporally symmetric processes must be non-signalling} 
	\label{figsig}
\end{figure}
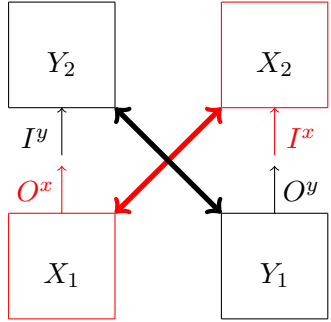

\section{Global Determinism \label{gd}}

 In the terminology of the process framework, global determinism requires that for any instance of a given process, the local controllables and the ontic variable together determine the outcomes. That is, for a given process, the ontic variable defines a function from the local controllables to the outcomes: 
 
 \begin{quotation} 
 	\hspace{-10mm}\textbf{Necessary condition for global determinism:} \emph{For any process $P$ with outcomes $O \in \mathscr{O}$,  there exists a set $V  \in \mathscr{V}$ consisting of local controllables and/or an ontic variable,  and a function $f^P: \mathscr{V} \rightarrow \mathscr{O}$ such that if  process $P$ is performed with the local controllables and ontic variable taking values $v \in \mathscr{V}$, then we obtain outcome variables $o \in \mathscr{O} $ such that $o  = f^P(v)$. }
 	\label{GD}
 	
 	\label{necessary}
 	
 \end{quotation}

 Suppose that  we have some process $P$ with an outcome which can be written as a single bit $b$.  How much information must be provided by the ontic variable $\Omega$ in order to make this process deterministic - i.e. to specify which outcome actually occurs? If we consider only a single instance of the process $P$, the answer is simple: either we need one bit (if the outcome is not a deterministic function of the local controllables) or we need zero bits (if the outcome \emph{is} a deterministic function of the local controllables). However, recall that the ontic variable includes global variables, and such global variables may describe not only individual instances of such processes but also sequences of processes spread across spacetime. So let us consider a sequence of $n$ instances of the  process $P$ and ask how much information must be provided by the ontic variable to make the full sequence of processes deterministic - i.e. to specify which sequence of outcomes actually occurs?
 
 To answer this question, we may employ well-known results from information theory concerning the compression of sequences.\cite{CoverThomas} If $b$  is uniformly random, then the results cannot be compressed and therefore the amount of information needed to make the full sequence of processes deterministic is simply $n$. But if $b$ is not random then the sequence can be compressed, since certain sequences are more likely than others, and we can choose to encode the more likely sequences using shorter encodings. The Shannon source coding theorem tells us that as $n$ goes to infinity, the number of bits required to specify which sequence actually occurs with negligble probability of error is $n H( b ) $, where $H(b)$ is the Shannon entropy of the random variable $b$.\cite{CoverThomas} That is, the total amount of `memory' that the universe must set aside  per instance of this process in order to make the outcomes deterministic is given by $H( b ) $, which lies in the range $[0, 1]$. Thus we will henceforth use Shannon entropies $H( ...)$ and Shannon mutual informations $I( : )$ to characterise the information involved in various sorts of processes, even though these quantities will not usually be integers. These entropies and mutual informations are to be understood as describing averages over a large number of repetitions of the relevant process (i.e. the probabilities used to calculate them can be interpreted within a frequentist framework).

 We also need an assumption about how processes behave under composition. We will make the simple assumption that when two processes are composed, the outcome of the composed process is a function of the remaining local controllables together with the ontic variables for the individual processes:
 
 \begin{definition} 
 	\textbf{Global determinism under composition:} Consider two processes $e_i : i \in \{ 0, 1\}$, with local controllables $ N_i$ taking values in $\mathscr{N}_i$, ontic variables $\Omega_i $ taking values in $\mathscr{W}_i$, and outcomes $O_i$ taking values in $\mathscr{O}_i$, for which it is possible to construct a composition $e_{12}^m$ in which the subset of variables $O_1^m \subseteq O_1$ are used to fix the subset of variables $N_2^m \subseteq N_2$ and the subset of variables $O_2^m \subseteq O_2$ are used to fix the subset of variables $N_1^m \subseteq N_1$.
 	
 	\vspace{2mm}
 	
 	If the world is globally deterministic, there exists a function $f^{e_1, e_2, m}: \{ (\mathscr{N}_1 \setminus \mathscr{N}_1^m) \otimes (\mathscr{N}_2 \setminus \mathscr{N}_2^m) \otimes \mathscr{W}_1 \otimes \mathscr{W}_2 \} \rightarrow \mathscr{O}_1 \otimes \mathscr{O}_2 $, such that when $e_{12}^m$ is performed with the local controllables equal to $n_1 \in (\mathscr{N}_1 \setminus \mathscr{N}_1^m)$ and $n_2 \in (\mathscr{N}_2 \setminus \mathscr{N}_2^m)$, and the ontic variables equal to $\Omega_1\in \mathscr{W}_1$ and $\Omega_2 \in \mathscr{W}_2$, we obtain outcomes $o_1 \otimes o_2 \in \mathscr{O}_1 \otimes \mathscr{O}_2$ such that $   o_1 \otimes o_2 = f^{e_1, e_2, m} (n_1 \otimes n_2 \otimes \Omega_1 \otimes \Omega_2)$.
 	
 	\label{note1}
 	
 \end{definition}

 Crucially, this means that there is no extra `memory allocation' available to choose the outcomes for the composed process - we must be able to determine the outcomes of the composed process using the same ontic variables as we would have used to determine the outcomes of the two processes individually.

	\subsection{Allowed Loops in a Deterministic Universe \label{signal}}

In ref \cite{QMG} we proved that subject to the assumptions set out in section \ref{gd}, in a globally deterministic universe all spatiotemporally symmetric processes must be non-signalling. Loosely speaking, this is because the amount of information provided by the ontic variable to determine the outcomes of the process is complementary to the amount of information provided by the local controllables, so if outcome $O$ depends too strongly on the local controllable $I$, then under the composition shown in fig \ref{figsig}, there will not be enough information in the ontic variables to fully determine the outcomes of the composed process, so global determinism will fail.

This example is particularly simple because box $1$ is defined as to have no input, or only the trivial input. But in general local labelled regions will have local controllables, and from condition \ref{necessary}, if the world is globally deterministic then for any given set of non-local boxes we can obtain from the ontic variable $\Omega$ and the function $f_P$ a   function $f_P^{\Omega}$ which maps all possible local controllables to the values of the outcomes that will occur if these local controllables  are chosen as our inputs.  Consider then the case where we have a bipartite process $P$, with Alice and Bob  selecting local controllables $a, m$ respectively and obtaining outcomes $c, g$ respectively. The ontic variable $\Omega$ defines a function $f^{\Omega}_{P} $  from two inputs (the two local controllables $a, m$) to two outputs (the two  outcomes $c, g$). We can also use $f^{\Omega}_P$ to define a function $f^{\Omega}_A$ which specifies which outcomes Alice would have obtained for every possible  choice of the local controllables $(a, m)$  i.e. $f^{\Omega}_A$ is a function from the inputs $a, m$ to the output $c$. By conditioning on different possible values for $m$ we can obtain a family of functions $f^{\Omega;m}_A$  which specify which outcomes Alice would have obtained for every possible  choice of the local controllable $a$ conditional on the value of $m$  i.e. each $f^{\Omega;m}_A$ is a function from the input $a$ to the output $c$. Given any $f^{\Omega;m}_A$ we can define a counterfactual outcome for  Alice's box, $\vec{c} :  c_j = f^{\Omega;m}_A(j)$ for some integer labelling  $\{ j\}$ of the set of possible values for the   local controllable $a$ - that is to say, the counterfactual outcome for a given box is a mapping from that box's input to that box's output, which can be regarded as a random vector-valued variable that may have nontrivial correlations with the inputs to other boxes.  

Condition \ref{note1} tells us how counterfactual outcomes behave under composition: it must be possible to define the counterfactual outcomes for a composed process using only the `memory allocation' required to define the counterfactual outcomes for the individual processes. Thus we can proceed with a proof very similar to that in ref \cite{QMG}. The underlying idea is exactly the same as in the no-signalling proof:  the amount of information provided by the ontic variable to determine the outcomes of the process is complementary to the amount of information provided by the local controllables, so if the \emph{counterfactual} outcome $\vec{g}$ depends too strongly on the local controllable $a$, then under the composition shown in fig \ref{figsig2}, there will not be enough information in the ontic variables to fully determine the \emph{counterfactual} outcomes of the composed process, so global determinism will fail.

\begin{theorem} \label{parind}
	Let $P$ be a spatiotemporally symmetric process which can be modelled by a pair of non-local boxes $1$, $2$ such that box $1$ takes an input $m$ and produces an outcome $g$, and box $2$ takes an input $a$ and produces an outcome $c$. Let $f^{\Omega}_2$ be some fixed value for the function $f^{\Omega}_2 :  \{ (m, a) \} \rightarrow \{ c \}$ specifying the outcome  that will be obtained from box $2$ for every possible choice of local controllables $a, m$. Then if the universe is globally deterministic, after conditioning on $f^{\Omega}_2$ we must have $ I (\vec{ g} : a  ) = 0$, where $\vec{g}$ is  the counterfactual outcome $\vec{g}$ for box $1$. 
\end{theorem}

	\begin{figure}
		\centering
		\begin{tikzpicture}[scale=0.7]

		\coordinate (a) at (0,0);
		\coordinate (b) at (0,2);
		\coordinate (c) at (2,2);
		\coordinate (d) at (2,0);

		\draw[red] (a) -- (b);
		\draw[red] (b) -- (c);
		\draw[red] (c) -- (d);
		\draw[red] (d) -- (a);

		\coordinate (e) at (4,0);
		\coordinate (f) at (4,2);
		\coordinate (g) at (6,2);
		\coordinate (h) at (6,0);

		\draw[black] (e) -- (f);
		\draw[black] (f) -- (g);
		\draw[black] (g) -- (h);
		\draw[black] (h) -- (e);
		
		\coordinate (i) at (4,4);
		\coordinate (j) at (4,6);
		\coordinate (k) at (6,6);
		\coordinate (l) at (6,4);

		\draw[red] (i) -- (j);
		\draw[red] (j) -- (k);
		\draw[red] (k) -- (l);
		\draw[red] (l) -- (i);
		
		\coordinate (m) at (0,4);
		\coordinate (n) at (0,6);
		\coordinate (o) at (2,6);
		\coordinate (p) at (2,4);

		\draw[black] (m) -- (n);
		\draw[black] (n) -- (o);
		\draw[black] (o) -- (p);
		\draw[black] (p) -- (m);

		\draw[red, ->] (1,2) -- (1,2.9);
		\draw[black, ->] (1,3.1) -- (1,4);
		
		\draw[black, ->] (5,2) -- (5,2.9);
		\draw[red, ->] (5,3.1) -- (5,4);
		
		\draw[black, ->] (1,6) -- (1,7.5);
		\draw[red, ->] (5,6) -- (5,7.5);
		
		\draw[red, ->] (1,-1.5) -- (1,0);
		\draw[black, ->] (5,-1.5) -- (5,0);
		
		\node[below ] at (0.5,7) {$c^y$};
		\node[below] at (5.5,-0.5) {$m^y$};
		\node[below, red] at (0.5,-0.5) {$m^x$};
		\node[below, red] at (5.5,7) {$c^x$};

		\node[below, red] at (0.5,2.8) {$g^x$};
		\node[below] at (5.5,2.8) {$g^y$};
		\node[below] at (0.5,3.8) {$a^y$};
		\node[below, red] at (5.5,3.8) {$a^x$};
		
		\draw[red, <->, line width=2pt] (c) -- (i);
		\draw[black, <->, line width=2pt] (p) -- (f);

		\node[below] at (1,1.25) {$X_1$};
		\node[below] at (5,5.25) {$X_2$};
		\node[below] at (1,5.25) {$Y_2$};
		\node[below] at (5,1.25) {$Y_1$};
		\end{tikzpicture}	
		\caption{Schematic diagram of the composition of two processes used in the proof of theorem \ref{parind}} 
		\label{figsig2}
	\end{figure}
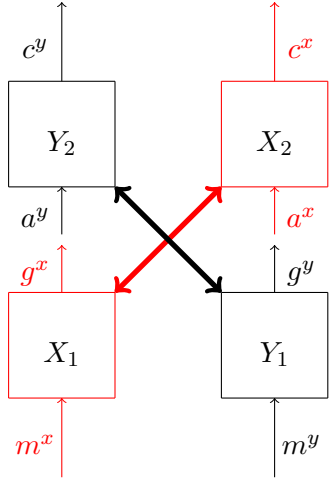

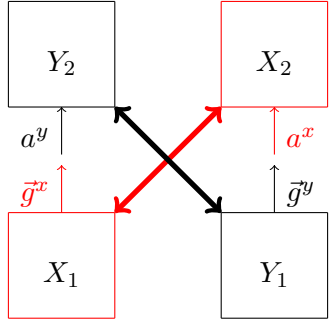
\begin{figure}
	\centering
	\begin{tikzpicture}[scale=0.7]

	\coordinate (a) at (0,0);
	\coordinate (b) at (0,2);
	\coordinate (c) at (2,2);
	\coordinate (d) at (2,0);

	\draw[red] (a) -- (b);
	\draw[red] (b) -- (c);
	\draw[red] (c) -- (d);
	\draw[red] (d) -- (a);

	\coordinate (e) at (4,0);
	\coordinate (f) at (4,2);
	\coordinate (g) at (6,2);
	\coordinate (h) at (6,0);

	\draw[black] (e) -- (f);
	\draw[black] (f) -- (g);
	\draw[black] (g) -- (h);
	\draw[black] (h) -- (e);
	
	\coordinate (i) at (4,4);
	\coordinate (j) at (4,6);
	\coordinate (k) at (6,6);
	\coordinate (l) at (6,4);

	\draw[red] (i) -- (j);
	\draw[red] (j) -- (k);
	\draw[red] (k) -- (l);
	\draw[red] (l) -- (i);
	
	\coordinate (m) at (0,4);
	\coordinate (n) at (0,6);
	\coordinate (o) at (2,6);
	\coordinate (p) at (2,4);

	\draw[black] (m) -- (n);
	\draw[black] (n) -- (o);
	\draw[black] (o) -- (p);
	\draw[black] (p) -- (m);

	\draw[red, ->] (1,2) -- (1,2.9);
	\draw[black, ->] (1,3.1) -- (1,4);
	
	\draw[black, ->] (5,2) -- (5,2.9);
	\draw[red, ->] (5,3.1) -- (5,4);

	\node[below, red] at (0.5,2.8) {$\vec{g}^x$};
	\node[below] at (5.5,2.8) {$\vec{g}^y$};
	\node[below] at (0.5,3.8) {$a^y$};
	\node[below, red] at (5.5,3.8) {$a^x$};
	
	\draw[red, <->, line width=2pt] (c) -- (i);
	\draw[black, <->, line width=2pt] (p) -- (f);

	\node[below] at (1,1.25) {$X_1$};
	\node[below] at (5,5.25) {$X_2$};
	\node[below] at (1,5.25) {$Y_2$};
	\node[below] at (5,1.25) {$Y_1$};
	\end{tikzpicture}	
	\caption{Simpler cyclic process after conditioning on a fixed value of the functions $f^{\Omega}_{ X_2}$,  $f^{\Omega'}_{ Y_2}$ and considering the counterfactual outcome $\vec{g}$ used for the proof of theorem \ref{parind}} 
	\label{figsig3}
\end{figure}

	\begin{proof}

		Since process $P$ is spatiotemporally symmetric, we can create two copies of this process: let the first copy be implemented by boxes  $X_1, X_2$ and the second instance y boxes $Y_1, Y_2$. Then these copies of $P$ can be used to construct a cyclic process as in figure \ref{figsig2}. Here, Alice takes boxes $X_1, Y_2$ to one spatiotemporal location and Bob takes $Y_1, X_2$ to a different location. At some time, Alice inputs a randomly chosen variable $m^x$ into $X_1$, obtains an outcome $g^x$, and then uses $g^x$ as the input $a^y$ to box $Y_2$, obtaining another outcome $c^y$. Likewise, at some time Bob inputs a randomly chosen variable $m_y$ into $Y_1$, obtains an outcome $g^y$, and then uses $g^y$ as the input $a^x$ to box $Y_2$, obtaining another outcome $c^x$.
		\vspace{2mm}
		 
		 We assume that the world is globally deterministic, so for each of these pairs of boxes there exists a function $f^{\Omega}_{X_1 X_2}$,  $f^{\Omega'}_{Y_1, Y_2}$ mapping values of the local controllables $a, m$ to outcomes $c, g$. From these functions we can define the functions $f^{\Omega}_{ X_2}$,  $f^{\Omega'}_{ Y_2}$ mapping values of the local controllables $a, m$ to the outcome $c$ for box $2$. Let us   select a subset of instances of these processes for which both of these functions $f^{\Omega}_{ X_2}$,  $f^{\Omega'}_{ Y_2}$ take some fixed form -  note that there's no need for them both to be the same. Within this set of processes the outputs of boxes $X_2, Y_2$ are  trivial and therefore we will  henceforth ignore them. We are interested in studying the behaviour of the counterfactual outcomes $\vec{g}^x$  and $\vec{g}^y$, and therefore we may also take it that the input to box $1$ is trivial (i.e. it is always the counterfactual input where we select all the local controllables at once). We then have a simpler cyclic process, as shown in figure \ref{figsig3}.
		
			\vspace{2mm} 
		
	From the composition rule (\ref{note1}), we have $H(\vec{g}^x \vec{g}^y | \Omega^x \Omega^y) = 0$, and hence $H( \vec{g}^x \vec{g}^y) = I(\vec{g}^x \vec{g}^y   :\Omega^x \Omega^y) $.

		\vspace{2mm} 
		
		From the definition of the mutual information, $H( \vec{g}^x \vec{g}^y ) =  H(\vec{g}^x) + H(\vec{g}^y)  - I( \vec{g}^x : \vec{g}^y)$ 
		
		\vspace{2mm} 
		In this construction, $I(\vec{g}^x : \vec{g}^y ) = I(\vec{g}^y : a^y) = I(\vec{g} : a)$, so $H(\vec{g}^y  \vec{g}^x) = 2H(\vec{g}) - I(\vec{g} : a)$ 
		
		\vspace{2mm} 
		
		Since $\vec{g}^x$ is a function of $a^x$ and $\Omega^x$  and $\vec{g}^y$ is a function of $a^y$ and $\Omega^y$, $I(\vec{g}^y  \vec{g}^x :\Omega^x \Omega^y) \leq I(a^y \vec{g}^y  : \Omega^y ) + I(a^x \vec{g}^x : \Omega^x) = 2 I (a \vec{g} : \Omega )$.
		
		\vspace{2mm} 
		
		Since $\vec{g}$ is a function of $a$ and $\Omega$, and $a$ and $\Omega$ are independent,  using lemma \ref{lemmasum} (see appendix \ref{lemma} for details) we have  $H(\vec{g} ) =  I(\vec{g} : a)  + I(a \vec{g} : \Omega)$. 
		\vspace{2mm} 
		
		Combining these results, we obtain: 
		
		\begin{equation}2 H(\vec{g}) - I(\vec{g} : a) \leq 2 H(\vec{g}) - 2 I (\vec{g} : a) \end{equation} 
		
		\vspace{3mm} 
		
		Hence $I(\vec{g} : a) \leq 0 $. But the Shannon mutual information is nonnegative, so  $I(\vec{g} : a) = 0$.  
		
		\end{proof}

The condition set out in  theorem \ref{parind} tells us that  that, conditional on a fixed choice of the function $f^{\Omega}_2$, which determines the output for box $2$ conditional on the inputs to both boxes, the counterfactual outcome $\vec{g}$ for box $1$ will never depend on which input is actually chosen for box $1$ - i.e. if we know which outcomes box $2$ will produce for any choice of inputs $a, m$, we can write the  output $g$ to box $1$   as a function of the input $m$ to box $1$ only, without any dependence on the input $a$ to box $2$.  This condition can therefore be understood as an analogue of parameter independence for counterfactual outcomes, so we will refer to it as   \emph{counterfactual parameter independence}. (See Appendix \ref{Bell} for an explanation of this terminology).

But wait - doesn’t this mean that the function $f^{\Omega}_2$ and the counterfactual outcome $\vec{g}$ could be used as a hidden variable $\lambda$ in the proof of Bell’s theorem, and therefore processes satisfying counterfactual parameter independence could never violate a Bell’s inequality? No, because although we have shown that $\vec{g}$ is independent of $a$ for a fixed value of $f^{\Omega}_2$, the function $f^{\Omega}_2$ determines the output of box $2$ conditional on both $a$ and $m$ , so the hidden variable $f^{\Omega}_2, \vec{g}$ does not provide enough information to determine the output of box $2$ without information about the  choice of input to box $1$. Therefore processes satisfying counterfactual parameter independence can still violate a Bell's inequality. 

But wait - here we are post-selecting on the output of box $2$. And it has been shown that if we allow post-selection then quantum mechanics does not even obey \emph{no-signalling}, so surely quantum mechanics with post-selection can’t obey counterfactual parameter independence, which is a stronger requirement than no-signalling? However this reasoning is incorrect, because in real experiments we can condition on at most one outcome, i.e. the outcome  for the measurement that Alice actually performs. Thus we can't post-select `the set of processes in which Alice would have obtained outcome $0$ if she had measured in basis $b$,' we can only post-select `the set of processes in which Alice did in fact measure in basis $b$ and she obtained outcome $0$.' So in ordinary post-selection we are conditioning in part on information about \emph{which measurement Alice actually makes} - no wonder it’s possible to violate no-signalling under these circumstances! Whereas the `post-selection’ involved in counterfactual parameter independence involves conditioning on a function specifying what outcome Alice \emph{would have} obtained to all of the possible measurements both she could have made, so this variable contains no information about which measurement Alice actually \emph{does} make. Obviously this is not a post-selection that we can ever actually perform in real life, because we can never know what result would have been obtained to all the counterfactual measurements that Alice did not actually make and therefore we can never select all those cases where the function $f^{\Omega}_{ A}$ has a certain value. However, if it is assumed that the world is globally deterministic, then there is always a fact of the matter about the function $f^{\Omega}_A$ for each instance of the process, and therefore the set of  instantiations of the process where the function $f^{\Omega}_A$ takes a certain form is well-defined, even if we ourselves can never identify that set. 
 
In the next section we will derive the Tsirelson bound from counterfactual parameter independence, but first it is important to show that quantum mechanics  satisfies counterfactual parameter independence, since we don't want to impose a constraint so strict that it rules out quantum mechanics itself! Of course, counterfactual parameter independence is a property of counterfactual outcomes, and standard quantum mechanics does not say anything about counterfactual outcomes, so we can't demonstrate this directly. Rather, we will demonstrate that for any bipartite quantum system, it is always possible to find a distribution over values of the counterfactual outcomes which correctly reproduces the quantum statistics and obeys counterfactual parameter independence.  This is done by choosing a temporal ordering for the measurements within the standard temporally ordered quantum formalism: we first `measure' one box and thus constructing a counterfactual outcome $\vec{g}$ for that box which is necessarily independent of the input to the second box, and then we subsequently `measure' the other box to create the function $f^{\Omega}_2$ mapping both inputs to the output for the second box. So if in fact there do exist definite counterfactual outcomes for quantum systems, they could obey counterfactual parameter independence. 

\begin{theorem} \label{qs}
	Whenever a quantum system can be modelled as a process $P$ involving a pair of non-local boxes	$1$, $2$ such that box $1$ takes an input $m$ and produces an outcome $g$, and box $2$ takes an input $a$ and produces an outcome $c$, it is always possible to construct  a distribution over values of the function $f^{\Omega}_P : \{ (m, a) \} \rightarrow \{ (g, c) \}$, which reproduces the quantum statistics and also obeys counterfactual parameter independence

	\end{theorem}

\begin{proof}

Consider a quantum system $\mathbb{X}_1 \mathbb{X}_2$ in the joint state $\rho$. We will perform a measurement in the set $\{ m_0, m_1 ...  \}$ on $\mathbb{X}_1$, obtaining an outcome in the set $\{ g_0, g_1 ... \}$; likewise we will perform a measurement in the set $\{ a_1, a_2 ... \}$ on $\mathbb{X}_2$, obtaining an outcome in the set $\{ c_0, c_1 ...\}$. This quantum system can be modelled as a pair of boxes $1$, $2$  such that box $1$ takes an input $m$ and produces an outcome $g$, and box $2$ takes an input $a$ and produces an outcome $c$.

\vspace{2mm} 

Consider an ensemble of such systems. We generate the functions  $f^{\Omega}_P$ for the ensemble as follows. First, for each instance of the process, for each possible measurement $m_i$ on $\mathbb{X}_1$, we generate an outcome $g^x$ by drawing from the possible outcomes according to the quantum probability distribution  $p(g^x) = Tr(  (G_i^x \otimes  \mathbb{I}_{\mathbb{X}_2} )   \rho )$, where  $G_i^x$ is the quantum measurement operator associated with obtaining outcome $g^x$ to the measurement $m_i$. Thus we have generated a counterfactual outcome $\vec{g}$ for the box $\mathbb{X}_1$ which is independent of $a$. 

Then, for each combination of measurements $m_i, a_j$ we generate an outcome $c^y$ for the measurement $a_j$ by first generating the post-measurement state for $\mathbb{X}_1$ associated with the chosen outcome $g^x$ of the measurement $m_i$ on $\mathbb{X}_1$, which is given by $\rho' =  Tr_{\mathbb{X}_2} (\frac{K_j^y \rho (K_j^y)^{\dagger} }{Tr( (G_i^x \otimes  \mathbb{I}_{\mathbb{X}_2} )   \rho ))})$ where $K_i^j$ is the Kraus operator associated with the operator $G_i^x \otimes  \mathbb{I}_{\mathbb{X}_2}$. Then we draw $c^y$  from the set of possible outcomes according to the quantum probability distribution $p(c^y) = Tr( C_j^y \rho')$ where $C_j^y$ is the quantum measurement operator associated with obtaining outcome $c^y$ to the measurement $a_j$. Clearly averaging over outcomes selected in this way will reproduce the quantum statistics.  

\vspace{2mm} 

Now for each member of the ensemble, we can use the function $f^{\Omega}_P$ to write down the function $f^{\Omega}_2$  mapping the inputs $a, m$ to an output $c$ for box $2$, and then we can select the subset of the ensemble where the function $f^{\Omega}_2$ is equal to some fixed function.  Since the procedure we have defined reproduces the usual quantum statistics, we can make any choice for the function $f^{\Omega}_2$ which assigns outcomes that are possible according to standard quantum mechanics. Then we can calculate the relative frequencies of the different values for $\vec{g}$ occurring within this subset of the ensemble, and thus define a probability distribution over the values of $\vec{g}$ which will reproduce the correct quantum statistics. We have not yet selected the actual value of the input $a$, and hence it is clear that this probability distribution does not depend on the value of $a$, i.e. conditional on any possible choice of $f^{\Omega}_2$ we will have $I(\vec{g} : a)$.

	\end{proof}

 \section{The Tsirelson Bound \label{Tsi}} 
 
In section \ref{entropy} we argued that information causality and the Tsirelson bound should be explained in a way that makes it clear why quantum systems sometimes behave like classical variables despite the existence of inequivalent measurements. We now see that global determinism can provide this explanation, because even though it is not possible to actually perform all measurements at once, in a globally deterministic universe the counterfactual outcome must still be well-defined and therefore it can be treated as a classical variable. Global determinism thus explains why the universe `cares' about counterfactual measurements as well as actual ones.

 In this spirit, we now show that information causality can be derived from counterfactual parameter independence. It follows that in a globally deterministic universe, information causality must always be obeyed, and since the Tsirelson bound follows from information causality we have shown that global determinism implies the Tsirelson bound.

 \begin{theorem}\label{ic}
 	For a spatiotemporally symmetric process $P$ which can be modelled by a pair of non-local boxes $1$, $2$ such that box $1$ takes an input $m$ which is an integer in the range $[ 0, n)$ and produces a one-bit outcome $g$, and box $2$ takes an input $A = a_0 a_1 ... a_{n - 1}$ which is a bit string of length $n$ and produces an outcome $c$, if the world is globally deterministic we must have:  
 	$ \sum_x I(g c : a_x | m = x) \leq H(c) $
 	\end{theorem}  
 
(Note that although this formulation looks highly specific, in fact any two-box process can be put in this form, as we can always ignore some bits if an outcome is too small or an input too large, and we can always repeat bits if an input is too small or an outcome too large.)

 \begin{proof}

 For the purpose of obtaining a contradiction, suppose that the world is globally deterministic, and that we have a spatiotemporally symmetric process $P$ which can be modelled as a pair of non-local boxes of the specified form satisfying:
 
 \[ I( g c : a_0 | m = 0) + I(g c : a_1 | m = 1) > H(c) \] 
 
As in the definition of counterfactual parameter independence, we condition on a fixed form for the function $f$ which maps the inputs $a, m$ to the output $c$; let us choose the function such that $c$ depends only on $a$ and not on $m$, so $c$  can be treated as a fixed classical variable. Then $I(g c : a_0 | m = 0) \leq I(\vec{g} c : a_0 )$ since $\vec{g}$ includes the outcome of the measurement $m = 0$, and likewise $I(g c : a_1 | m = 1) \leq I(\vec{g} c : a_1)$, so we have:

\begin{align*} I( g c : a_0 | m = 0) + I(g c : a_1 | m = 1) 
\\ \leq  I( \vec{g} c : a_0 ) + I(\vec{g} c : a_1 )  \\
 \leq  I(\vec{g} c : a_0 )  + I(\vec{g} c a_0 : a_1  )  \\
 = I(\vec{g} c : a_0 ) + I(\vec{g} c : a_1  | a_0) + I(a_0 : a_1) 
\\   = I(\vec{g} c : a_0 a_1 )  + I(a_0 : a_1)  
\\= I(\vec{g} c : a_0 a_1 ) 
\\= I (\vec{g} c : A )\\ \end{align*}

where at the third line we have used the strong subadditivity of the Shannon mutual information, and at the fourth and fifth lines we have used the chain rule for the Shannon mutual information. Recall that  $c, \vec{g}$ and $a_0, a_1$ are all simply classical variables, so these classical rules are valid here. At the sixth line we have used the fact that $a_0$ and $a_1$ are independent.

Combining these results and using the chain rule again, we infer that: 

 \[ I(c :A | \vec{g} ) +  I( \vec{g} :A ) > H(c) \]

But from the definition of the mutual information and the non-negativity of the (classical) conditional entropy, we have that $I(c : A | \vec{g}) < H(c)$. So we conclude that, for fixed $c$: 

\[ I( \vec{g} : A) > 0 \] 

That is to say, counterfactual parameter independence is violated, and from theorem \ref{parind}, this means that the world is not globally deterministic. 

We have obtained a contradiction;  so if the world is globally deterministic, every spatiotemporally symmetric process of this form must satisfy: 
 
  \[ I( g c : a_0 | m = 0) + I(g c : a_1 | m = 1) \leq H(c) \] 
  
  We can then proceed to add further terms $I( g c : a_x | m = x)$ and for each term follow the same reasoning to finally obtain: 
  
    \[ \sum_x  I( g c : a_x | m = x)  \leq H(c) \]

\end{proof}

\subsection{Example: PR boxes} 

The PR boxes are a set of non-local boxes proposed by Popescu and Rohrlich which exhibit correlations that are non-signalling but also not allowed by quantum mechanics. In particular, the PR boxes violate information causality. Here we will use the PR boxes to give an example of the connection between the violation of counterfactual parameter independence and the violation of information causality.

The PR boxes are a pair of boxes $1$, $2$ taking bit valued inputs $m, a$ respectively and producing bit valued outcomes $g, o_2$ respectively. The results are guaranteed to satisfy the equality $g \oplus o_2 = m a$ (using bitwise addition). To violate information causality with these boxes, we convert a two bit input $A = (a_0, a_1)$ into a one bit input $a = a_0 \oplus a_1$ and use it as the input  to box $2$. We then add $a_0$ to the output $o_2$ of box $2$, producing a new outcome $c$. Then if $m = 0$ the results are guaranteed to satisfy $g \oplus c = a_0$ and if $m = 1$ the results are guaranteed to satisfy $g \oplus c = a_1$ - i.e. we have $g \oplus c = a_{m}$. Then we have $\sum_x I( a_x : g c | m = x ) = 2 > H(c) = 1$, so information causality is indeed violated.

It is straightforward to see that no possible distribution over counterfactual outcomes for the PR boxes could ever satisfy counterfactual parameter independence.  To see this, let us condition on a fixed function $f^{\Omega}_2$ for our second box. In particular, let us choose a function where the output $c$ depends only on $a$ and not on $m$, i.e. for fixed $a$ we will get the same $c$ for any value of $m$. Note that for fixed $a_0, a_1, c$ we have that $g_{m = 0} = a_0 \oplus c$ and $g_{m = 1} = a_1 \oplus c$, and therefore if it were possible to know both $g_{m = 0} $ and $g_{m = 1} $, whenever we found it to be the case that $g_{m = 0}  \oplus g_{m = 1}  = 0$ we could infer that $a_0 = a_1$. This demonstrates that after conditioning on this fixed function $f^{\Omega}_2$, the counterfactual outcome $\vec{g}$ for the first box can't be independent of the input $a = a_0, a_1$ to the second box.  

It may be helpful to see why the proof of theorem \ref{qs} does not apply to the PR boxes. With a set of PR boxes, we can of course follow the  procedure set out in the proof, generating a counterfactual outcome for box $1$ and then constructing the  function $f^{\Omega}_2$  for box $2$ conditional on choice of counterfactual outcome for the first box. However, for certain choices of the counterfactual outcome for the box $1$, some possible choices for the function $f^{\Omega}_2$ will be impossible, even though they represent physically possible outcomes. For example, if we select a counterfactual outcome $\vec{g}$ such that $g_{m = 0}  \neq g_{m = 1} $ then we cannot select a function $f^{\Omega}_2$ which maps $ a = x y, m = 0 \rightarrow c = 0$ and $a = x y, m = 1 \rightarrow c = 0$ for any $x, y$, even though this combination of inputs and outputs is  physically possible,  because in that case we would not have both  $g_{m = 0} = a_0 \oplus c$ and $g_{m = 1} = a_1 \oplus c$. Thus the construction of theorem \ref{qs} can't be applied in this case because averaging over its results would not reproduce the expected joint statistics for the PR boxes.

\section{Monogamy of Correlations \label{monogamy}}

Quantum mechanics has a property known as `monogamy of entanglement,' which refers to the fact that the amount of entanglement a quantum system has with a second system limits the amount of entanglement it can have with any third system - in the most extreme case, a quantum system which is maximally entangled with a second quantum system cannot be entangled at all with any other quantum system. But entanglement is a property of quantum states, usually quantified using state-dependent measures like the `concurrence' (or `tangle')\cite{Koffmanetal}, and these measures do not have a   straightforward physical interpretation. On the other hand, we can also characterize this feature in operational language by talking about the `monogamy of correlations,' which is a property of certain theories in which the strength of the correlations  that a system has with a given second system limits the strength of its correlations   with any third system. The monogamy of correlations is defined purely in terms of the $CHSH$ quantity, which is entirely constructed from observable statistics for measurements on a pair of systems $A, B$, so it is easily expressed in purely operational language. This allows us to derive general statements about the monogamy of correlations in various possible theories.

 For example, it has been shown that any operationally defined theory which does not allow superluminal signalling  must obey the following  monogamy bound for measurements on three systems $1$, $2$ and $3$: \cite{Toner}

\begin{equation}
\label{NS}
CHSH_{13} + CHSH_{23}  \leq 4   \end{equation} 

The CHSH quantities used here are as defined in section \ref{entropy}. Note that this bound holds only if the two measurements on system $3$ used to calculate $CHSH_{13}$ are the same as the two measurements on $3$ used to calculate $CHSH_{23}$. 

It has also been shown that  all \emph{quantum} correlations obey the stronger monogamy bound\cite{Toner2} (again, using the same two measurements on system $3$ for both CHSH quantities): 

\begin{equation}
\label{Q}
CHSH_{13}^2   + CHSH_{23}^2  \leq 8  \end{equation}

The gap between the no-signalling monogamy bound and the quantum monogamy bound is reminiscent of the gap between the no-signalling bound and the Tsirelson bound for the CHSH quantity, which suggests the possibility that the two gaps might be explained in the same way. Indeed, it is possible to derive the Tsirelson bound as a special case of the monogamy bound \ref{Q}, since clearly this equation entails that neither $CHSH_{13}$ nor $CHSH_{23}$ can be greater than $2 \sqrt{2}$. For example, refs \cite{Toner2, Seevinck2} derive the quantum monogamy bound from the mathematical structure of quantum mechanics and thence derive the Tsirelson bound.

However, our aim here is to understand features of quantum mechanics without presupposing its mathematical structure, and therefore we would like to work in the other direction: having derived the Tsirelson bound from global determinism, we would like to use it to obtain the quantum monogamy bound. We have seen that no-signalling and information causality are general constraints which must be obeyed by any spatiotemporally symmetric processes in a globally deterministic universe; and thus if the world is globally deterministic, any possible spatiotemporally symmetric composition of boxes must also satisfy no-signalling and information causality. In appendix \ref{monogamy2}, we show that the quantum monogamy bound can be derived from this requirement, thus demonstrating that the quantum monogamy bound would also follow from global determinism:

 \begin{theorem} \label{themonogamy}
 	
 		For a spatiotemporally symmetric process $P$ which can be modelled by a set of non-local boxes $1$, $2$, $3$, if these boxes obey information causality and no-signalling,  we must have:
 		\[ CHSH_{13}^2 + CHSH_{23}^2  \leq 8 \] 
 		
 		where $CHSH_{13}$ is the $CHSH$ quantity for the boxes $1,3$, and $CHSH_{23}$ is the $CHSH$ quantity for the boxes $2,3$, and we use the same two measurements on box $3$ for both $CHSH_{13}$ and $CHSH_{23}$

\end{theorem}

\section{The Horizon Problem}

There exists a famous problem in theoretical physics known as `the horizon problem.' Here we use the results of section \ref{monogamy} to sketch a novel approach to solving the horizon problem. We will not go into the detailed physics of the black hole here; it is enough for our purposes to describe the problem in simple conceptual terms.  We leave it to future work to describe how this toy model relates to  current knowledge about the the quantum field theoretic description of black holes.

 It is believed that black holes emit Hawking radiation which results from pair creation, i.e. the creation of two particles in a pure, maximally entangled state, one falling into the black hole and the other escaping. Consider a Hawking pair $B, C$ emitted at time $t$ after the Page time, with $B$  being the outgoing particle and $C$ the infalling particle, and denote by $A$ the
collection of all Hawking radiation emitted earlier than $t$.
Theoretical arguments show that $B$ must be entangled with $A$; but since $B$ is already maximally entangled with $C$, this would entail a violation of the monogamy bounds \ref{NS} and \ref{Q}\cite{Grudka_2018}. 

Many solutions to this problem involve the suggestion that the entanglement is broken when the particle passes the black hole horizon\cite{Almheiri_2013}, but the results of section \ref{monogamy} suggest a different option. Suppose we accept that the monogamy of correlations holds as a consequence of global determinism - that is, quantum correlations are required to be monogamous  in order to avert the possibility of creating a closed causal loop via the construction depicted in fig \ref{figsig2}. Note that this argument applies both to the no-signalling monogamy bound or the quantum monogamy bound: each of these bounds is hypothetically violated by the black hole, but we have argued that the former holds as a consequence of no-signalling and the latter as a consequence of no-signalling together with information causality, and that both no-signalling and information causality are required to hold precisely in order to prevent closed causal loops.

 But as we have noted, the construction of closed causal loops requires that the process in question be spatiotemporally symmetric, i.e. that it can be operated in both directions - and the correlations produced by the black hole are \emph{not} spatiotemporally symmetric! $C$ is maximally entangled with $A$ and $B$, which are both outside the black hole and so could theoretically be brought together and measured simultaneously, allowing signalling from $C$ to $A, B$;  but the mechanism does not allow us to create a converse construction with  $A', B'$ inside the black hole and $C'$ outside, so we cannot also signal in the opposite direction, and therefore this particular violation of monogamy does not allow the construction of a closed causal loop. So, if it is indeed the case that the quantum monogamy bound holds only as a consequence of global determinism, then the monogamy of correlations does \emph{not} need to hold in this instance, since a violation of monogamy here does not entail a violation of global determinism. This offers a simple solution to the horizon problem - perhaps $B$ really is maximally entangled with both $A$ and $C$, since the monogamy bounds \ref{NS} and \ref{Q} don't hold in this case.

\section{Contextuality}

One particularly counterintuitive property of quantum mechanics is the fact that it is \emph{contextual}. In its first incarnation, contextuality was associated with the idea that measurement outcomes correspond to definite properties of systems - so if a system has the property associated with measurement outcome $C$, then when we perform a measurement which includes $C$ as a possible outcome, we are certain to get outcome $C$. A theory in which this is true can be said to obey \emph{deterministic non-contextuality}. It is trivial to find an assignation of definite properties satisfying deterministic non-contextuality in a setting where each measurement outcome can occur in only one measurement, but of course in quantum mechanical systems of more than two dimensions, a given measurement outcome  can occur in several different possible measurement; and the Kochen-Specker theorem proves that for certain collections of quantum-mechanical measurements in greater than two dimensions there is no possible way of picking a set of outcomes $S$ which correspond to `definite properties'\cite{KochenSpecker} - if we insist on assigning deterministic outcomes, we will sometimes have to assign them in such a way that we are certain to get outcome $C$ when we perform measurement $A$, but certain to \emph{not} get outcome $C$ when we perform some other measurement $A'$ which also has $C$ as a possible outcome. So quantum mechanics does not obey deterministic non-contextuality and therefore it is not possible to think of quantum measurement outcomes as simply describing pre-existing properties of systems.

One might worry that the failure of deterministic non-contextuality means that global determinism cannot possibly be true. However, in fact there is no contradiction between the failure of deterministic non-contextuality and global determinism, because in our presentation of global determinism we have  assumed that the measurement results prescribed by ontic variables are a function of the local controllables, i.e. the set of measurements performed. Thus it is perfectly consistent with this approach that an ontic variable might prescribe that we will obtain outcome $C$ if we perform measurement $A$ but we will not obtain outcome $C$ if we perform some other measurement $A'$ which also has $C$ as a possible outcome. Indeed, since we assume that the measurement results prescribed by determinism are a function of \emph{all} the local controllables for a given experiment, it is even possible that the global variable might prescribe that we will obtain outcome $C$ if we perform measurement $A$ and another experimenter performs measurement $M$, but we will \emph{not} obtain outcome $C$ if we perform measurement $A$ and the other experimenter performs some other measurement $M'$. So deterministic contextuality is not only unproblematic but in fact entirely to be expected on this account. 

There is another approach to this subject, due to Spekkens, which identities the motivating principle behind non-contextuality as the idea that operationally equivalent situations should correspond to the same underlying ontic reality\cite{Spekkenscontextuality};  an ontological model for a theory satisfies \emph{Spekkens non-contextuality} iff any situations which are operationally equivalent are represented identically in the ontological model, where situations are described as `operationally equivalent' whenever they give rise to exactly the same measurement statistics. In particular, we say an ontological model obeys \emph{preparation non-contextuality} iff  two preparations which prepare the same quantum state always result in the same probability distribution over ontic states. Spekkens showed that no ontological model which reproduces all the results of quantum mechanics can be preparation non-contextual\cite{Spekkenscontextuality}, thus demonstrating that quantum mechanics exhibits a form of contextuality with or without the assumption of determinism.\footnote{Spekkens also showed that no ontological model which reproduces all the results of quantum mechancis can be measurement non-contextual, but we will not discuss measurement contextuality here.}

It might be tempting to suppose that we can get around this result by observing that the Spekkens proof is based on the assumption that measurement results can depend only on facts about the past, whereas in a global setting measurement results can depend on facts about the present and future as well, so there is no reason to assume that the preparation contextuality result would apply in a global context. However, Costa and Shrapnel have used the process matrix formalism to demonstrate that even if we allow non-standard temporal orderings, no model can explain quantum correlations from non-contextual ontological properties of the world, be they initial states, dynamical laws, or global constraints\cite{Shrapnel_2018}. The Costa-Shrapnel proof applies to ontic variables such as those we have considered here, and therefore a globally deterministic account of quantum mechanics would necessarily exhibit contextuality at the level of the ontic variables. 

This is not  a death knell for the approach, since there is nothing logically inconsistent about the possibility of a universe governed by a contextual theory, and indeed all the evidence so far suggests that our universe is such a universe. However, one may still find contextuality aesthetically displeasing, or worry that it indicates an incompleteness in the theory; so let us say a few words about how preparation contextuality might be realised in this context, in the hope that this might alleviate some of the discomfort.

In the globally deterministic picture, for  any given prepare-measure scenario there is a fixed value of the ontic variable $\Omega^{PM}$ which specifies  what the measurement outcome will be for any choice of preparation $P$ and measurement $M$. Thus, conditional on performing preparation $e_X$ there is a fixed ontic variable $\Omega^M$ specifying what the measurement outcome will be for any possible measurement. Denote by $\Omega^M(g | M_A)$ the result $g$ which will definitely be obtained if we perform measurement $M_A$ when the ontic variable is equal to $\Omega^M$. In general the value of the  ontic variable is different each time we perform $e_X$, so we can define a distribution  $\mu_X(\Omega^M)$, which expresses the relative frequency of ontic variables $\Omega^{PM}$ which after conditioning on $e_X$ yield the ontic variable $\Omega^M$. Thus if we perform preparation $e_X$ followed by measurement $M_A$, the probability that we get result $r$ is $\sum_{\Omega^M : \Omega^M(g | M_A) = r} \mu_X(\Omega^M)$. Clearly for any preparation $e_Y$ which according to standard quantum mechanics prepares the same state as $e_X$, we must have that for any $M_A, r$,  $\sum_{\Omega^M : \Omega^M(g | M_A) = r} \mu_X(\Omega^M) = \sum_{\Omega^M : \Omega^M (g | M_A) = r} \mu_Y(\Omega^M)$, since the two preparations must produce identical distributions over measurement outcomes. 

But the distributions over \emph{counterfactual} measurement outcomes do not have to be equal. For example, it could perfectly well be the case that when we perform $e_X$, the outcome prescribed by the ontic variable for measurement $M_A$ is always the same as the outcome prescribed for some incompatible measurement $M_B$, whereas when we perform $e_Y$, the outcome prescribed by the ontic variable for measurement $M_A$ is always different from the outcome prescribed for $M_B$. This will make no difference to the observable statistics associated with $e_X$ and $e_Y$, since we can never perform both $M_A$ and $M_B$ on a single system and therefore we will never know that the distributions are different. So $e_X$ and $e_Y$ can be operationally equivalent whilst not being equivalent at the level of \emph{counterfactual} outcomes. 

Since a key lesson of sections \ref{Tsi} and \ref{monogamy} was that to explain quantum behaviour we must take counterfactual outcomes seriously as elements of reality, this analysis suggests that standard preparation contextuality is simply not the right tool to apply in a globally deterministic theory. If the counterfactual outcomes are elements of reality, then the notion of equivalence that matters is not straightforward operational equivalence but rather \emph{counterfactual operational equivalence}, which is exhibited by any two situations which lead to the same statistics over \emph{counterfactual} outcomes. Then we can say that an ontological model obeys \emph{counterfactual non-contextuality} iff any situations which exhibit counterfactual operational equivalence are represented identically in the ontic model. And since the `ontic states' in the globally deterministic picture are just the ontic variables which specify the counterfactual outcomes, two situations giving rise to the same distributions over counterfactual outcomes must also be represented identically in the ontic model. So global determinism is compatible with the type of non-contextuality that is most relevant for the global setting, although it cannot satisfy ordinary preparation contextuality.

It's also important to keep in mind that there are no objective probabilities in a globally deterministic universe: for each instance of a given process the ontic variable has a fixed value determined by facts about the configuration of the rest of the universe, although we ourselves don't have access to all those facts and thus don't know the value of the ontic variable. The probability distributions $\mu(\Omega)$ are therefore \emph{epistemic} probabilities encoding our uncertainty about the global facts which determine the global variable. So although two different ways of preparing the same quantum state may give rise to different distributions $\mu(\Omega)$, this does not mean that they `correspond to different underlying ontic realities,' since there are no probability distributions at the ontic level, just a single value of the ontic variable. Thus the motivating concern behind Spekkens contextuality seems less relevant in the globally deterministic case.

\section{Conclusion} 

There is a small but growing chorus of voices in theoretical physics arguing that quantum mechanics should be understood as the local limit of a global theory. In this paper we have continued the project of attempting to understand  what local physics has to tell us about the structure of the global laws. In particular, we aimed to show that several prima facie puzzling features of quantum mechanics can be explained by appeal to the hypothesis that the universe is deterministic on a global scale whilst appearing probabilistic at the level of the variables that we ourselves can discover and manipulate. The gap between the quantum and no-signalling non-locality bounds and the gap between the quantum and no-signalling monogamy bounds can both be explained in this way, and there may even be a new route to solving the black hole horizon problem. 

Obviously none of this constitutes conclusive proof for the hypothesis of global determinism; inference to the best explanation is abductive rather than deductive reasoning, and there will inevitably be legitimate disagreements about the merits of various possible explanations. But we suggest that the explanation put forward here is at the very least doing something importantly different from explanations of the Tsirelson bound that previously been advanced. Information causality, while suggestive, seems too specific and agent-centric to be a candidate for a fundamental principle of nature, and although  there exist various entropic derivations of information causality, these approaches seem to merely restate in a more mathematical language the fact that quantum measurements outcomes are not quite as non-classical as they could have been. We are left with the question of \emph{why} quantum systems should sometimes behave like classical variables even though it is not possible to perform all possible measurements simultaneously, and for the scientific realist, the answer must  tell a clear story about how this behaviour arises from general features of an underlying reality.  The assumption of global determinism provides such an answer:  even though it is not possible to perform all possible measurements simultaneously, determinism requires that the universe can always say definitely which outcome would occur for each possible measurement.  Therefore what we have referred to as a `counterfactual outcome' must be regarded as an element of reality despite its inaccessibility to local observers, so it \emph{can} be regarded as a classical variable under appropriate circumstances. The existence of this classical variable explains why quantum mechanics does not explore the full limits of the possibilities opened up by the existence of incompatible measurements. 

There are many interesting directions for further investigation in this area. Although we have derived the exact bound on quantum correlations for the case of the CHSH quantity, it is known that information causality is still not adequate to fully define the set of quantum correlations in more general setups. Indeed, thus far no general principle of this kind has been able to distinguish between the set of quantum correlations and the set of \emph{almost-quantum} correlations defined in ref \cite{Navascu_s_2015}. It would certainly be interesting if global determinism could be used to distinguish the two sets. Moving beyond the study of non-local correlations, a logical next step would be to derive the whole of quantum mechanics from hypotheses about the global structure of the laws of nature. There are already a number of derivations of quantum mechanics from general operational or information-theoretic principles\cite{Hardyreasonable, Masanes_2013, CDK, Bub}, but these derivations are subject to the same criticism as we have levelled at the entropic derivations of information causality: they are suggestive, but they don't provide a convincing realist story about why the classical and quasi-classical principles they invoke should still hold in this non-classical theory. Thus there is still room for a new derivation on realist grounds similar to that we have given in this paper.  A complementary line of investigation involves attempting to write down specific global models from which the local behaviour of quantum mechanics can be derived. If the arguments of this paper are accepted, these models should be deterministic in a suitable global sense, which may help to constrain the search space for suitable models. Finally, the possibility of global, all-at-once laws of nature raises many deep conceptual questions about our usual ways of thinking about time, probability, determinism and modality, and there is a lot of work to be done  untangling these consequences and building a new conceptual framework for the global approach.

\appendix

\section{Costa and Shrapnel \label{cs}}

Our generalized ontological model framework is based on the formalism defined by Costa and Shrapnel in ref \cite{Shrapnel_2018}, although Costa and Shrapnel work in terms of an \emph{ontic process} $\omega$, rather than the \emph{ontic variable} $\Omega$ that we have used here.  Costa and Shrapnel define the ontic process to be invariant  under local operations, but then they allow that the ontic process may depend on the environment $W$, i.e. they define a function $p(\omega | W)$ which is presumably intended to be non-trivial. But the environment \emph{can} be controlled by local operations - by choosing to prepare our entangled particles in one state or another prior to measuring them, we are making a choice of environment -  so this terminology could make it look as if the ontic process is not invariant under local operations after all. 

In fact, we presume that what Costa and Shrapnel mean is that there exists a fixed background variable $\Omega$ which can't be changed by local operations, and for any given experiment, $\Omega$ contains a set of instructions $\omega$ conditional on each possible choice of the environment $W$, so $p(\omega | W)$ can be understood as `the relative frequency of  background variables $\Omega$ which define instructions $\omega$ conditional on choosing environment $W$.' That is, the ontic variable does not depend on the choice of process, and the ontic process is derived directly from the ontic variable by conditioning on the choice of process - the value that the ontic process would have conditional on each possible choice of environment is already fixed prior to our choice of environment, so the ontic variable is not being altered by the choice of process. 

Thus the sake of clarity, in this article we have chosen to work in terms of ontic variables rather than ontic processes, so it is entirely clear that the ontic variable $\Omega$ is independent of all local operations, including the choice of environment and/or process. This also means that the ontic variable does not depend on where we put the `cut' between environment and local controllables. This is important, because the methodology that we adopt in the proof of theorem \ref{parind} requires that the ontic variable can't under any circumstances be made to depend on the outcomes of other processes, and therefore we need a definition which ensures that the ontic variable is invariant under \emph{all} possible local operations.

\section{Lemma \label{lemma}}

We need the following lemma for the proof in section \ref{signal}.

\begin{lemma} 
	If  $N$ and $Q$ are statistically independent and $O$ is a function of $N$ and $Q$, then $H(O ) =  I(O : N)  + I(N O : Q)$
	\label{lemmasum} 
\end{lemma}

\begin{proof} 
	
	Since $O$ is a function of $N$ and $Q$, $H(O | N  Q) = 0$, and thus from Bayes' rule  $H(N  Q | O) + H(O) = H(N Q)$ 
	
	\vspace{2mm}
	
	Using the chain rule for conditional entropy, we obtain $H(N | O) + H(Q | N O) + H(O) = H(N Q) $ 
	
		\vspace{2mm}

	Since $N$ and $Q$ are statistically independent,  $H(N Q) = H(N) + H(Q)$
	
		\vspace{2mm}

	Thus $ H(O) = H(N) - H(N| O) + H(Q) - H(Q | NO) = I(O : N)  + I(N  O : Q)$
	
\end{proof} 

\section{Parameter and Outcome Independence  \label{Bell}}

	In the proof of Bell's theorem the assumption of `locality' is encoded mathematically by the assumption of factorisability. That is, we assume there exists an `ontic state' $\lambda$ which mediates the correlations between the measurement outcomes, so that: 
	
	\begin{equation}p (s, t | \lambda, M^A = i, M^B = j) = p^A(s | \lambda, M^A = i) p^B(t | \lambda, M^B = j)  \end{equation}
	
	As pointed out by Jarrett, factorisability may be understood by the product of two separate assumptions\cite{Jarrett}: 
	
	\begin{enumerate}
		
		\item \textbf{Parameter independence:} The probability distribution over the outcomes of the measurement performed by experimenter $B$ does not depend on which measurement experimenter $A$ chooses to perform, and vice versa: 
		
		\begin{equation} p^B (t | \lambda, M^A = i, M^B = j) = p^B (t | \lambda, M^B = j)   \end{equation}

		\item \textbf{Outcome independence:} The probability distribution over the outcomes of the measurement performed by experimenter $B$ does not depend on the outcome of the measurement performed by experimenter $A$, and vice versa. \cite{Shimony}:

		\begin{equation} p^B (t | \lambda, M^B = j, O^A = s) = p^B (t | \lambda, M^B = j)   \end{equation}

	\end{enumerate}

	The violation of Bell's inequality by quantum mechanics therefore indicates that either parameter independence or outcome independence must be violated in nature, but not necessarily both.

	These definitions for parameter independence and outcome independence are of course based on the assumption that we have a \emph{local} hidden variable $\lambda$. But one could also define parameter independence and outcome independence for a more general context. For a pair of measurements made by experimenters $A$ and $B$, let $\Lambda$ be a hidden variable which specifies the probability for the set of outcomes $s, t$ conditional on the measurement choices $M^a = i, M^b = j$. ($\Lambda$ could just prescribe the quantum probabilities, or it could go beyond the quantum probabilities - in the most extreme case, it could prescribe the outcomes as a deterministic function of the measurement choices). 
	
	In a general non-local setting, $\Lambda$ may prescribe that the outcome of $B$'s measurement depends on $A$'s choice of measurement and/or the result of $A$'s measurement. But we can still meaningfully distinguish between dependence on the parameter and dependence on the measurement. In this context, it would seem reasonable to say that (generalised) parameter independence holds if, after conditioning on $\Lambda$ and on the outcome of the measurement performed by experimenter $A$, there is no further dependence of the outcomes of the measurement performed by $B$ on the choice of the measurement $A$: the nonlocal dependence can all be attributed to the outcome of $A$. Similarly, it would seem reasonable to say that (generalised) outcome independence holds if, conditional on $\Lambda$ and on the choice of measurement performed by experimenter $A$, the distribution over the outcomes of the measurement performed by experimenter $B$ is independent of the outcome of the measurement performed by experimenter $A$. This explains our use of the terminology `(counterfactual) parameter independence' - by conditioning on a fixed outcome, we can determine whether the nonlocal correlations require a dependence on the parameter as well the outcome. 
	
\section{Proof for section \ref{monogamy} \label{monogamy2}}
 \begin{theorem} \label{themonogamy}
 	
 		For a spatiotemporally symmetric process $P$ which can be modelled by a set of non-local boxes $1$, $2$, $3$, if these boxes obey information causality and no-signalling,  we must have:
 		\[ CHSH_{13}^2 + CHSH_{23}^2  \leq 8 \] 
 		
 		where $CHSH_{13}$ is the $CHSH$ quantity for the boxes $1,3$, and $CHSH_{23}$ is the $CHSH$ quantity for the boxes $2,3$, and we use the same two measurements on box $3$ for both $CHSH_{13}$ and $CHSH_{23}$

\end{theorem}

\begin{proof} 
	
	Consider a set of three boxes $B_1,B_2,B_3$ such that boxes $B_1, B_2$ accept inputs $m_1, m_2$ and produce outputs $g_1, g_2$, and box $B_3$ accepts input $a$ and produces output $c$.  
As shown in ref \cite{Masanes_2006}, without changing the value of any CHSH quantities we can always adjust the boxes such that for each pair $i, 3$ the probability of obtaining $g_i \oplus c = a m_i$ is the same for both $m_i = 0$ and $m_i = 1$ and is equal to $P_i = \frac{1}{2}(1 + e_i)$, so we will henceforth work entirely with boxes in this form. We will henceforth suppose that box $B_1$ is correlated with $B_3$ and box $B_2$ is correlated with $B_3$ but $B_1$ and $B_2$ are not correlated with one another, i.e. the complete information about $m_1, a, c$ screens off $g_1$ from $m_2$ and vice versa.   If this is not the case we can always employ local randomisation procedures to remove any correlations between $B_1$ and $B_2$. 

Now suppose Alice has box $3$ and Bob has boxes $1, 2$. Given an input $m$, Bob selects corresponding values $m_1, m_2$ by some deterministic procedure, inputs these values into his boxes, and then performs some post-processing on the outputs $g_2, g_2$ to produce an output $g$, thus constructing a composite box $B_{12}$. Let the probability of obtaining $g \oplus c = am$ using this composite box be equal to $P_{12} = \frac{1}{2}(1 + e_{12})$. 

Note that we have:

\begin{multline} 
I(g  : a  m | c  ) =  I( g_1 g_2 : a m | c  )   = I( g_1 : a  m  | c ) + I( g_2 : a m | c ) - I(g_1 : g_2 | c  ) \\
=   I(g_1 : a m_1 | c  ) + I(g_2 : a  m_2 | c ) - I(g_1 : g_2 | c  )  \\  
\end{multline} 

where at the second line we have used the fact that the complete information about $m_1, a, c$ screens off $g_1$ from $m_2$ and vice versa.

\vspace{2mm} 

Note that no-signalling entails that if $e_2 >  0$  we must have $p(g_1 = g_2) = \frac{1}{2}$ for each value of $c$. Otherwise, we could violate the no-signaling principle by the following procedure.  First consider the case $p(g_1 = g_2) > \frac{1}{2}$. Alice has boxes $B_1, B_3$ and Bob has box $B_2$. Alice makes a guess $g_1 = g_2$, which is correct with probability $p(g_1 = g_2) = \frac{1}{2}(1 + e_{g_1 g_2})$. Alice chooses input $a = 1$ for box $B_3$, then guesses  $m_2 = g_1 \oplus c$; this guess is correct with probability $(\frac{1}{2} + \frac{1}{2} e_{g_1 g_2})( \frac{1}{2} +\frac{1}{2}  e_2) + (\frac{1}{2} - \frac{1}{2} e_{g_1 g_2})( \frac{1}{2} - \frac{1}{2}  e_2) = \frac{1}{2} + \frac{1}{2} e_2  e_{g_1 g_2}$. This quantity is greater than $\frac{1}{2}$ since ex hypothesi $e_2$ and $ e_{g_1 g_2}$ are both greater than zero; so we have achieved signalling from Bob to Alice. For the case $p(g_1 = g_2) < \frac{1}{2}$ we proceed similarly except that Alice guesses $g_1 \neq g_2$ and thus again we achieve signalling from Alice to Bob. So the no-signalling principle entails that $p(g_1 = g_2) = \frac{1}{2}$ for each value of $c$. The same argument can be made if $e_1 > 0$, and hence provided that $e_2 > 0$ or $e_1 > 0$ (and we may assume that one of these is the case, since otherwise the monogamy bound is trivial), we have  $I(g_1 : g_2 | c ) = 0$.

\vspace{2mm}

Using the channel capacity formula, we have that that for $i \in \{ 1, 2\} \\ I( g_i : a  m_i | c) =   1 - h(P_i)$ and similarly $I( g : a m | c ) = 1 - h(P_{12})$ and hence: 
 
	\begin{align*} 
1 - h(P_{12}) = 2 - h(P_1) - h(P_2) \\
\therefore 1 - h(\frac{1}{2}(1 +  e_{12} )) =   2 - h(\frac{1}{2}(1 +  e_1 ))- h(\frac{1}{2}(1 + e_2 ))  \\ 
\therefore \frac{1}{2 \ln(2)} \sum_{q = 1}^{\infty}\frac{1}{q( 2q - 1)} e_{12}^{2q} = \frac{1}{2 \ln(2)} \sum_{q = 1}^{\infty}\frac{1}{q( 2q - 1)} (e_1^{2q} + e_2^{2q}     ) \\
\therefore \frac{1}{2 \ln(2)} (e_{12}^{2} + \sum_{q = 2}^{\infty}\frac{1}{q( 2q - 1)} e_{12}^{2q}    )= \frac{1}{2 \ln(2)} (e_1^{2} + e_2^{2}   +  \sum_{q = 2}^{\infty}\frac{1}{q( 2q - 1)} (e_1^{2q} + e_2^{2q} )    ) \\ 
\therefore  e_{12}^{2}    = e_1^{2} + e_2^{2} +  s     \\ 
\end{align*} 
 
 where $s = \sum_{q = 2}^{\infty}\frac{1}{q( 2q - 1)} (e_1^{2q} + e_2^{2q} -  e_{12}^{2q})$;  note that $s < 0$ if $e_{12}^{2}    \geq  e_1^{2} +  e_2^{2}$.
 
Now recall that the protocol of ref \cite{Pawlowski} achieves a value for the information causality quantity which is equal to $\sum^n_{k = 0} { n \choose k}( 1- h(\frac{1 + e_{m = 0} ^{n - k} e_{m = 1}^k}{2}))$, using a  set of $n$ pairs of boxes  such that the probability of obtaining $g \oplus c = a m$ is given by $ \frac{1}{2}(1 +  e_{m = 0})$ when $m = 0$ and $ \frac{1}{2}(1 +  e_{m = 1})$. Imagine that Alice and Bob create $n$ copies of the pairs $B_{12}, B_3$ and use them to execute the protocol exhibited in ref \cite{Pawlowski}. Thus  $e_{m = 0} = e_{m = 1} = e_{12}$, so this procedure achieves a value for the information causality quantity given by:

	\begin{align*} 
I = 	\sum^n_{k = 0}{ n \choose k}( 1- h(\frac{1 + (\sqrt{ e_1^2 + e_2^2 + s})^2}{2}))  \\
	=  \frac{1}{2 \ln(2)} \sum^n_{k = 0}{ n \choose k}\sum_{q = 1} \frac{1}{q( 2q - 1)} (e_1^2 + e_2^2 + s)^{qn} \\ 	\end{align*} 
	
	Now note that:

	\begin{align*} (e_1^2 + e_2^2 + s)^n + \frac{(e_1^2 + e_2^2 + s)^{2n}}{6} \geq  (e_1^2 + e_2^2 )^n - | s|^n  + \frac{(e_1^2 + e_2^2)^{2n}}{6}  -\frac{ |s|^{2n}}{6}
 \\=   (e_1^2 + e_2^2 )^n - | \sum_{q = 2}^{\infty}\frac{1}{q^n( 2q - 1)^n} (e_1^{2q} + e_2^{2q} -  e_{12}^{2q})^n |  + \frac{(e_1^2 + e_2^2)^{2n}}{6}  -\frac{ |s|^{2n}}{6} \end{align*} 
 
 But for sufficiently large $n$ we will always have $| \sum_{q = 2}^{\infty}\frac{1}{q^n( 2q - 1)^n} (e_1^{2q} + e_2^{2q} -  e_{12}^{2q})^n |  \leq \frac{(e_1^2 + e_2^2)^n}{6}$, so we obtain $(e_1^2 + e_2^2 + s)^n + \frac{(e_1^2 + e_2^2 + s)^{2n}}{6} \geq  (e_1^2 + e_2^2 )^n  -\frac{ |s|^{2n}}{6}$. The negative term $-\frac{ |s|^{2n}}{6}$ can be absorbed into the positive part of the next term in the same way, and so on throughout the summation, so we obtain:

 \begin{align*} 
I   \geq  \frac{1}{2 \ln(2)} \sum^n_{k = 0}{ n \choose k} (e_1^2 + e_2^2 )^{n} \\
	= (2 e_1^2 + 2e_2^2) ^n 
\end{align*} 
 
	So the information causality quantity will exceed $1$ for some sufficiently large $n$ if $2 e_1^2 + 2e_2^2 > 1$; thus we must have   $e_1^{2} + e_2^{2} \leq \frac{1}{2}$. Then since $CHSH_{13} = 4 e_1$ and $CHSH_{23} = 4 e_2$,we finally obtain: 

\[  CHSH_{13}^2 + CHSH_{23}^2    \leq 8 \] 
\end{proof}

\section{Multipartite Information Causality \label{Hsu}}

We are aware of an alternative derivation of the quantum monogamy bound from information causality, in ref  \cite{Hsu_2012}. This paper defines a generalisation of information causality known as `multipartite information causality,' and suggests that the quantum monogamy bound can be derived from it. However, this cannot be correct, as neither classical nor quantum mechanics actually obey `multipartite information causality,' as defined here.

Multipartite information causality is defined in ref \cite{Hsu_2012} as follows: 

\begin{definition} 
	
		For a set of $q$ non-local boxes $1$, $2$ ...  $q$ such that box $i$ for $i \in [1 ... q - 1)$ takes an input $m_i$ which is an integer in the range $[ 0, n)$ and produces a one-bit outcome $g_i$, and box $1$ takes an input $\vec{A} = a_0 a_1 ... a_{n - 1}$ which is a bit string of length $n$ and produces an outcome $c$ such that $I( c: A) = 0$, multipartite information causality requires: 
	
	\[\sum_i \sum_w  I(a_w : c g_i | m_i = w)  \leq H(c) \]

	\end{definition} 

To see that multipartite information causality does not hold in quantum mechanics, consider the case $q = 3$ and suppose that the three boxes are constructed from a GHZ state $\frac{1}{\sqrt{2}}(| 0 0 0 \rangle + | 1 1 1 \rangle)$. All three parties perform a measurement in the basis $\{ | 0 \rangle , | 1 \rangle\}$ producing three outcomes $g_x = g_y = o_z$. We then construct $c = a_0 \oplus o_z$. Since $o_z$ is locally random, it is clear that $I( c: A) = 0$ and $H(c) = 1$. But also we have $g_x \oplus c = a_0$ and $g_y \oplus c = a_0$, so clearly, if we label $\{ | 0 \rangle , | 1 \rangle\}$ as measurement direction $0$, we have  $I(a_0 : c g_x | m_x = 0) = 1$ and also $I(a_0 : c g_y | m_y = 0) = 1$, so $\sum_w ( I(a_w : c g_x | m_x = w) +  I(a_x : c g_y | m_y = w)  ) \geq 2 > H(c) $. 

In fact, we do not actually need any non-locality to achieve this result - it is enough that the three parties share a classical random bit $b = g_x = g_y = o_z$. So `multipartite information causality' does not even hold in classical physics. 

The problem is that we are `double counting' information in the summation - we can obtain one bit of information about $A$ from box $1$ and one bit from box $2$, but we are obtaining the \emph{same} information from both of them, so although the summation comes to more than $H(c)$, in fact the total distinct information gained is no greater than $H(c)$. This indicates that the  correct version of multipartite information causality, for the case of three boxes only, is as follows: 

\begin{definition} 
	
	For a set of non-local boxes $1$, $2$, $3$ such that boxes $1$ and $2$ take  inputs $m_1, m_2$ which are integers in the range $[ 0, n)$ and produce one-bit outcomes $g_1, g_2$, and box $3$ takes an input $\vec{A} = a_0 a_1 ... a_{n - 1}$ which is a bit string of length $n$ and produces an outcome $c$ such that $I( c: A) = 0$, multipartite information causality requires: 
	
	\[ \sum_w ( I(a_w : c g_1 | m_1 = w) +  I(a_2 : c g_2 | g_1, m_1 = w, m_2 = w)  ) \leq H(c) \]

\end{definition} 

This version follows directly from standard information causality, as by applying the chain rule the constraint can be written as $\sum_w ( I(a_1 a_2 : c g_1 g_2 | m_1 = w, m_2 = w)  \leq H(c) $ and then $m_1 = w, m_2 = w$ can be treated as a single measurement with outcome $g_1, g_2$.  

 \bibliographystyle{plainurl} 

\bibliography{newlibrary11}{}

\begin{thebibliography}{10}

\bibitem{QMG}
Emily {Adlam}.
\newblock {Quantum Mechanics and Global Determinism}.
\newblock {\em Quanta}, 7:40--53, 2018.
\newblock URL: \url{http://quanta.ws/ojs/index.php/quanta/article/view/76}.

\bibitem{adlamspooky}
Emily Adlam.
\newblock {Spooky Action at a Temporal Distance}.
\newblock {\em Entropy}, 20(1):41, 2018.

\bibitem{Al_Safi_2011}
Sabri~W. Al-Safi and Anthony~J. Short.
\newblock Information causality from an entropic and a probabilistic
  perspective.
\newblock {\em Physical Review A}, 84(4), Oct 2011.
\newblock URL: \url{http://dx.doi.org/10.1103/PhysRevA.84.042323}, \href
  {https://doi.org/10.1103/physreva.84.042323}
  {\path{doi:10.1103/physreva.84.042323}}.

\bibitem{Almheiri_2013}
Ahmed Almheiri, Donald Marolf, Joseph Polchinski, Douglas Stanford, and James
  Sully.
\newblock An apologia for firewalls.
\newblock {\em Journal of High Energy Physics}, 2013(9), Sep 2013.
\newblock URL: \url{http://dx.doi.org/10.1007/JHEP09(2013)018}, \href
  {https://doi.org/10.1007/jhep09(2013)018}
  {\path{doi:10.1007/jhep09(2013)018}}.

\bibitem{Barnum_2010}
Howard Barnum, Jonathan Barrett, Lisa~Orloff Clark, Matthew Leifer, Robert
  Spekkens, Nicholas Stepanik, Alex Wilce, and Robin Wilke.
\newblock Entropy and information causality in general probabilistic theories.
\newblock {\em New Journal of Physics}, 12(3):033024, Mar 2010.
\newblock URL: \url{http://dx.doi.org/10.1088/1367-2630/12/3/033024}, \href
  {https://doi.org/10.1088/1367-2630/12/3/033024}
  {\path{doi:10.1088/1367-2630/12/3/033024}}.

\bibitem{Bellinequality}
J.~{Bell}.
\newblock On the {Einstein}-{Podolsky}-{Rosen} paradox.
\newblock {\em Physics}, 1:195--200, July 1964.

\bibitem{Bellfree}
J.~{Bell}.
\newblock Free variables and local causality.
\newblock In {\em Speakable and unspeakable in quantum mechanics}, chapter~12.
  Cambridge University Press, 1987.

\bibitem{Bub}
J.~{Bub}.
\newblock {Why the quantum?}
\newblock {\em eprint arXiv:quant-ph/0402149}, February 2004.
\newblock \href {http://arxiv.org/abs/quant-ph/quant-ph/0402149}
  {\path{arXiv:quant-ph/quant-ph/0402149}}.

\bibitem{CDK}
Giulio Chiribella, Giacomo~Mauro D’Ariano, and Paolo Perinotti.
\newblock Informational derivation of quantum theory.
\newblock {\em Physical Review A}, 84(1), Jul 2011.
\newblock URL: \url{http://dx.doi.org/10.1103/PhysRevA.84.012311}, \href
  {https://doi.org/10.1103/physreva.84.012311}
  {\path{doi:10.1103/physreva.84.012311}}.

\bibitem{Tsirelson}
Boris~S Cirel'son.
\newblock Quantum generalizations of bell's inequality.
\newblock {\em Letters in Mathematical Physics}, 4(2):93--100, 1980.

\bibitem{Koffmanetal}
V.~{Coffman}, J.~{Kundu}, and W.~K. {Wootters}.
\newblock {Distributed entanglement}.
\newblock {\em Physical Review A}, 61(5):052306, May 2000.
\newblock \href {http://arxiv.org/abs/quant-ph/quant-ph/9907047}
  {\path{arXiv:quant-ph/quant-ph/9907047}}, \href
  {https://doi.org/10.1103/PhysRevA.61.052306}
  {\path{doi:10.1103/PhysRevA.61.052306}}.

\bibitem{Dahlsten_2012}
Oscar C~O Dahlsten, Daniel Lercher, and Renato Renner.
\newblock Tsirelson’s bound from a generalized data processing inequality.
\newblock {\em New Journal of Physics}, 14(6):063024, Jun 2012.
\newblock URL: \url{http://dx.doi.org/10.1088/1367-2630/14/6/063024}, \href
  {https://doi.org/10.1088/1367-2630/14/6/063024}
  {\path{doi:10.1088/1367-2630/14/6/063024}}.

\bibitem{D_Ariano_2016}
Giacomo~Mauro D’Ariano.
\newblock Physics without physics.
\newblock {\em International Journal of Theoretical Physics}, 56(1):97–128,
  Nov 2016.
\newblock URL: \url{http://dx.doi.org/10.1007/s10773-016-3172-y}, \href
  {https://doi.org/10.1007/s10773-016-3172-y}
  {\path{doi:10.1007/s10773-016-3172-y}}.

\bibitem{Goldstein_2003}
Sheldon Goldstein and Roderich Tumulka.
\newblock Opposite arrows of time can reconcile relativity and nonlocality.
\newblock {\em Classical and Quantum Gravity}, 20(3):557–564, Jan 2003.
\newblock URL: \url{http://dx.doi.org/10.1088/0264-9381/20/3/311}, \href
  {https://doi.org/10.1088/0264-9381/20/3/311}
  {\path{doi:10.1088/0264-9381/20/3/311}}.

\bibitem{Grudka_2018}
Andrzej Grudka, Michael J.~W. Hall, Michał Horodecki, Ryszard Horodecki,
  Jonathan Oppenheim, and John~A. Smolin.
\newblock Do black holes create polyamory?
\newblock {\em Journal of High Energy Physics}, 2018(11), Nov 2018.
\newblock URL: \url{http://dx.doi.org/10.1007/JHEP11(2018)045}, \href
  {https://doi.org/10.1007/jhep11(2018)045}
  {\path{doi:10.1007/jhep11(2018)045}}.

\bibitem{Holevo1973SomeEO}
A.~Holevo.
\newblock Some estimates of the information transmitted by quantum
  communication channels.
\newblock 1973.

\bibitem{10.3389/fphy.2020.00139}
Sabine Hossenfelder and Tim Palmer.
\newblock Rethinking superdeterminism.
\newblock {\em Frontiers in Physics}, 8:139, 2020.
\newblock URL:
  \url{https://www.frontiersin.org/article/10.3389/fphy.2020.00139}, \href
  {https://doi.org/10.3389/fphy.2020.00139}
  {\path{doi:10.3389/fphy.2020.00139}}.

\bibitem{Hsu_2012}
Li-Yi Hsu.
\newblock Multipartite information causality.
\newblock {\em Physical Review A}, 85(3), Mar 2012.
\newblock URL: \url{http://dx.doi.org/10.1103/PhysRevA.85.032115}, \href
  {https://doi.org/10.1103/physreva.85.032115}
  {\path{doi:10.1103/physreva.85.032115}}.

\bibitem{Jarrett}
Jon~P. Jarrett.
\newblock On the physical significance of the locality conditions in the bell
  arguments.
\newblock {\em Noûs}, 18(4):569--589, 1984.
\newblock URL: \url{http://www.jstor.org/stable/2214878}.

\bibitem{KochenSpecker}
Simon Kochen and E.P. Specker.
\newblock The problem of hidden variables in quantum mechanics.
\newblock In C.A. Hooker, editor, {\em The Logico-Algebraic Approach to Quantum
  Mechanics}, volume~5a of {\em The University of Western Ontario Series in
  Philosophy of Science}, pages 293--328. Springer Netherlands, 1975.
\newblock URL: \url{http://dx.doi.org/10.1007/978-94-010-1795-4_17}, \href
  {https://doi.org/10.1007/978-94-010-1795-4_17}
  {\path{doi:10.1007/978-94-010-1795-4_17}}.

\bibitem{PuseyLeifer}
M.~{Leifer} and M.~{Pusey}.
\newblock {Is a time symmetric interpretation of quantum theory possible
  without retrocausality?}
\newblock {\em ArXiv e-prints}, July 2016.
\newblock \href {http://arxiv.org/abs/1607.07871} {\path{arXiv:1607.07871}}.

\bibitem{Masanes_2013}
L.~Masanes, M.~P. Muller, R.~Augusiak, and D.~Perez-Garcia.
\newblock Existence of an information unit as a postulate of quantum theory.
\newblock {\em Proceedings of the National Academy of Sciences},
  110(41):16373–16377, Sep 2013.
\newblock URL: \url{http://dx.doi.org/10.1073/pnas.1304884110}, \href
  {https://doi.org/10.1073/pnas.1304884110}
  {\path{doi:10.1073/pnas.1304884110}}.

\bibitem{Masanes_2006}
Ll. Masanes, A.~Acin, and N.~Gisin.
\newblock General properties of nonsignaling theories.
\newblock {\em Physical Review A}, 73(1), Jan 2006.
\newblock URL: \url{http://dx.doi.org/10.1103/PhysRevA.73.012112}, \href
  {https://doi.org/10.1103/physreva.73.012112}
  {\path{doi:10.1103/physreva.73.012112}}.

\bibitem{Navascu_s_2015}
Miguel Navascués, Yelena Guryanova, Matty~J. Hoban, and Antonio Acín.
\newblock Almost quantum correlations.
\newblock {\em Nature Communications}, 6(1), Feb 2015.
\newblock URL: \url{http://dx.doi.org/10.1038/ncomms7288}, \href
  {https://doi.org/10.1038/ncomms7288} {\path{doi:10.1038/ncomms7288}}.

\bibitem{Oreshkov2}
O.~{Oreshkov}, F.~{Costa}, and {\v C}.~{Brukner}.
\newblock {Quantum correlations with no causal order}.
\newblock {\em Nature Communications}, 3:1092, October 2012.
\newblock \href {http://arxiv.org/abs/quant-ph/1105.4464}
  {\path{arXiv:quant-ph/1105.4464}}, \href {https://doi.org/10.1038/ncomms2076}
  {\path{doi:10.1038/ncomms2076}}.

\bibitem{Oreshkov}
O.~{Oreshkov} and C.~{Giarmatzi}.
\newblock {Causal and causally separable processes}.
\newblock {\em New J. Phys}.

\bibitem{Pawlowski}
M.~{Pawlowski}, T.~{Paterek}, D.~{Kaszlikowski}, V.~{Scarani}, A.~{Winter}, and
  M.~{{\.Z}ukowski}.
\newblock {Information causality as a physical principle}.
\newblock {\em Nature}, 461:1101--1104, October 2009.
\newblock \href {http://arxiv.org/abs/quant-ph/0905.2292}
  {\path{arXiv:quant-ph/0905.2292}}, \href
  {https://doi.org/10.1038/nature08400} {\path{doi:10.1038/nature08400}}.

\bibitem{Priceretro}
H.~{Price}.
\newblock {A Neglected Route to Realism About Quantum Mechanics}.
\newblock {\em ArXiv General Relativity and Quantum Cosmology e-prints}, June
  1994.
\newblock \href {http://arxiv.org/abs/gr-qc/9406028}
  {\path{arXiv:gr-qc/9406028}}.

\bibitem{PRICE_1994}
H.~Price.
\newblock A neglected route to realism about quantum mechanics.
\newblock {\em Mind}, 103(411):303–336, 1994.
\newblock URL: \url{http://dx.doi.org/10.1093/mind/103.411.303}, \href
  {https://doi.org/10.1093/mind/103.411.303}
  {\path{doi:10.1093/mind/103.411.303}}.

\bibitem{Rohrlich}
D.~{Rohrlich} and S.~{Popescu}.
\newblock {Nonlocality as an axiom for quantum theory}.
\newblock {\em eprint arXiv:quant-ph/9508009}, August 1995.
\newblock \href {http://arxiv.org/abs/quant-ph/quant-ph/9508009}
  {\path{arXiv:quant-ph/quant-ph/9508009}}.

\bibitem{Seevinck2}
M.~{Seevinck}.
\newblock {Monogamy of Correlations vs. Monogamy of Entanglement}.
\newblock {\em ArXiv e-prints}, August 2009.
\newblock \href {http://arxiv.org/abs/quant-ph/0908.1867}
  {\path{arXiv:quant-ph/0908.1867}}.

\bibitem{Shimony}
Abner Shimony.
\newblock Events and processes in the quantum world.
\newblock In Roger Penrose and C.~J. Isham, editors, {\em Quantum Concepts in
  Space and Time}, pages 182--203. New York ;Oxford University Press, 1986.

\bibitem{Shrapnel_2018}
Sally Shrapnel and Fabio Costa.
\newblock Causation does not explain contextuality.
\newblock {\em Quantum}, 2:63, May 2018.
\newblock URL: \url{http://dx.doi.org/10.22331/q-2018-05-18-63}, \href
  {https://doi.org/10.22331/q-2018-05-18-63}
  {\path{doi:10.22331/q-2018-05-18-63}}.

\bibitem{Spekkenscontextuality}
R.~W. {Spekkens}.
\newblock {Contextuality for preparations, transformations, and unsharp
  measurements}.
\newblock {\em Physical Review A}, 71(5):052108, May 2005.
\newblock \href {http://arxiv.org/abs/quant-ph/quant-ph/0406166}
  {\path{arXiv:quant-ph/quant-ph/0406166}}, \href
  {https://doi.org/10.1103/PhysRevA.71.052108}
  {\path{doi:10.1103/PhysRevA.71.052108}}.

\bibitem{CoverThomas}
J.~{Thomas} and T.~{Cover}.
\newblock {\em Elements of Information Theory}.
\newblock Wiley, 2006.

\bibitem{timpson2008philosophical}
Chris Timpson.
\newblock Philosophical aspects of quantum.
\newblock {\em The Ashgate companion to contemporary philosophy of physics},
  page 197, 2008.

\bibitem{Toner}
B.~{Toner}.
\newblock {Monogamy of non-local quantum correlations}.
\newblock {\em Proceedings of the Royal Society of London Series A},
  465:59--69, January 2009.
\newblock \href {http://arxiv.org/abs/quant-ph/quant-ph/0601172}
  {\path{arXiv:quant-ph/quant-ph/0601172}}, \href
  {https://doi.org/10.1098/rspa.2008.0149} {\path{doi:10.1098/rspa.2008.0149}}.

\bibitem{Toner2}
B.~{Toner} and F.~{Verstraete}.
\newblock {Monogamy of Bell correlations and Tsirelson's bound}.
\newblock {\em eprint arXiv:quant-ph/0611001}, November 2006.
\newblock \href {http://arxiv.org/abs/quant-ph/quant-ph/0611001}
  {\path{arXiv:quant-ph/quant-ph/0611001}}.

\bibitem{Whartoninformation}
K.~{Wharton}.
\newblock {Quantum States as Ordinary Information}.
\newblock {\em ArXiv e-prints}, March 2014.
\newblock \href {http://arxiv.org/abs/1403.2374} {\path{arXiv:1403.2374}}.

\bibitem{Wheeler1989-WHEIPQ}
John~Archibald Wheeler.
\newblock Information, physics, quantum: The search for links.
\newblock In {\em Proceedings III International Symposium on Foundations of
  Quantum Mechanics}, pages 354--358. 1989.

\end{thebibliography}

\end{document}